\documentclass[11pt]{article}
\usepackage[utf8]{inputenc}
\usepackage{hyperref}
\usepackage{mathtools}
\usepackage{amsmath}
\usepackage{amsthm}
\usepackage{amssymb}
\usepackage{authblk}
\usepackage{cite}
\usepackage{lipsum}
\usepackage{geometry}
\usepackage{enumerate}
\usepackage{array}
\usepackage{multicol}
\usepackage{multirow}
\usepackage{graphicx}
\usepackage{xcolor}
\usepackage{tikz-cd}
\usepackage{accents}

\usepackage{mathrsfs}

\usepackage[center]{caption}
\newtheorem{theorem}{Theorem}[section]
\numberwithin{theorem}{section}
\newtheorem{lemma}[theorem]{Lemma}

\newtheorem{definition}[theorem]{Definition}
\newtheorem{corollary}[theorem]{Corollary}
\newtheorem{prop}[theorem]{Proposition}
\newtheorem{remark}[theorem]{Remark}
\newtheorem{example}[theorem]{Example}
\DeclareMathOperator{\tr}{Tr}

\DeclareMathOperator{\nr}{Nr}

\newcommand{\Sum}{\displaystyle \sum}
\newcommand{\F}{\mathbb F}

\newcommand{\bit}{\epsilon}

\newcommand{\Fbn}{\mathbb{F}_{2^n}}

\newcommand{\Fbnn}{\mathbb{F}_{2^n} \longrightarrow \mathbb{F}_{2^n}}

\newcommand{\gm}{\gamma}

\newcommand{\cl}{\mathcal{L}}

\begin{document}

    \title{On non-monomial APcN permutations over \\
    finite fields of even characteristic}
    \author{Jaeseong Jeong$^1$, Namhun Koo$^2$, Soonhak Kwon$^1$\\
        \small{\texttt{ Email: wotjd012321@naver.com, nhkoo@ewha.ac.kr, shkwon@skku.edu}}\\
        \small{$^1$Applied Algebra and Optimization Research Center, Sungkyunkwan University, Suwon, Korea}\\
        \small{$^2$Institute of Mathematical Sciences, Ewha Womans University, Seoul, Korea}\\
    }

      \date{}
    \maketitle
    \begin{abstract}
    \vspace{0.5em}
Recently, a new concept called the $c$-differential uniformity was
proposed by Ellingsen et al. (2020), which allows to simplify some
types of differential cryptanalysis. Since then, finding functions
having low $c$-differential uniformity has attracted the attention
of many researchers. However it seems that, at this moment, there
are not many non-monomial permutations having low $c$-differential
uniformity. In this paper, we present new classes of (almost)
perfect
 $c$-nonlinear non-monomial permutations over a binary field.
\end{abstract}

\section{Introduction}

In block ciphers, S-boxes with good cryptographic properties play
important roles to resist several attacks. Recently, Ellingsen et.
al. \cite{EFR+20} proposed a new cryptographic tool,
$c$-Difference Distribution Table (and the corresponding
$c$-differential uniformity), which allows to simplify a new type
differential cryptanalysis proposed in  \cite{BCJW02}. In
\cite{EFR+20}, the $c$-differential uniformities of known PN
functions and the multiplicative inverse function were
investigated. Independently about the same time,  Bartoli and
Timpanella \cite{BT20} introduced the notion of $\beta$-planar
functions which are in fact PcN functions.

After the notion of $c$-differential uniformity was proposed,
finding functions with low $c$-differential uniformity has
attracted the attention of many researchers. In \cite{HPR+21}, the
authors proposed PcN power functions for $c=-1$ in odd
characteristic.
Several PcN power functions and functions with low
$c$-differential uniformity were studied in
\cite{MRS+21,WZ21,ZH21}.    In \cite{Yan22}, $c$-differential
uniformity of ternary APN power functions over a cubic field is
studied. Many classes of PcN and APcN functions were proposed in
\cite{WLZ21} by using the cyclotomic technique, switching method
and AGW criterion. Very recently,  several classes of PcN and APcN
functions were proposed in \cite{LRS22} by using generalized
Dillon's switching method.  In  \cite{BCRS22},  the authors
studied  $c$-differential uniformity of  some piecewise defined
functions.

 For a binary field, a new class of  APcN power
functions was proposed in  \cite{TZJT21}.  Also, in \cite{HPS+22},
the authors studied $c$-differential and boomerang uniformities of
two classes of permutation polynomials over a binary field.
 In Table 1, we summarize the
known PcN and APcN permutations over a binary field. In Table 2,
we give a list (though may not be complete) of PcN and APcN
functions over general finite fields, which are constructed from
other PcN or APcN functions.

Since permutations are very useful in constructing invertible
S-boxes, it is desirable to find permutations having good
cryptographic properties such as low $c$-differential uniformity.
However it seems that there are not many non-monomial APcN
permutations over a binary field, except for the ones in
\cite{HPS+22, WLZ21}. In this paper,  we present three classes of
non-monomial APcN permutations and a class of non-monomial PcN
permutations over $\Fbn$.


The rest of this paper is organized as follows. In Section $2$, we
give some notations and preliminaries. In Section $3$, we present
a new class of cubic APcN permutations. In Section $4$,  we show
that the compositional inverse of the APcN permutation proposed in
Section $3$ is also APcN. We also study the $c$-differential
uniformity of the involution proposed in Section $3$ and it's
generalization. Finally, in Section $5$, we give the concluding
remarks.

\begin{table}[h]
    \caption{The known PcN and APcN functions over  $\Fbn$}
    \renewcommand{\arraystretch}{1.3}
      \footnotesize
    \begin{tabular}{|c|c|c|c|c|}
        \hline
        $F(x)$ &  $_c\Delta_F$ & Permutation & Conditions & References \\
        \hline\hline
        $x^{2^n-2}$ &    $2$    &  $\bigcirc$  & $\tr(c) = \tr(1/c) = 1$, $c \notin \F_{2}$ &\cite{EFR+20}\\
           \hline
      $x+\tr(\alpha x+x^{2^i+1})$ &    $\leq 2$    &  $\bigcirc$  & $\tr(\alpha)=1$, $n>1$ is odd, $\gcd(n,i)=1$, $c \notin \F_{2}$ & \cite{HPS+22}\\
        \hline
       $x^{2^i+1} $ &    $1$    &  $\bigcirc$   & $\frac{n}{d}>1$ is odd, $d=\gcd(n,i)$, $c \in \F_{2^d}\setminus \F_2$   & \cite{MRS+21}\\
        \hline
        $x^{q^3+q^2+q-1}$ & 2 & $\bigcirc$ & $q = 2^i$, $n = q^4$, $c \in \mu_{q^2+1} \setminus \{1\}$ & \cite{TZJT21}\\
\hline
        \multirow{2}*{$(x+\tr(x^{2^i+1}))^{2^i+1}$} &  \multirow{2}*{2} &  \multirow{2}*{$\bigcirc$} & $d=\gcd(n,i)$ is even,
         &  \multirow{2}*{Theorem \ref{maintheorem}}\\
         &       &  &  $\frac{n}{d}>1$ is odd, $c \in \F_{2^d}\setminus \F_2$  & \\        \hline
            \multirow{2}*{$x^\frac{2^{in/d}+1}{2^i+1}+\tr(x)$} &  \multirow{2}*{2}
                                            &  \multirow{2}*{$\bigcirc$} &$d=\gcd(n,i)$ is even,
                                            &  \multirow{2}*{Theorem \ref{inversemaintheorem}}\\
                                           &       &  &  $\frac{n}{d}>1$ is odd,  $c \in \F_{2^d}\setminus \F_2$ & \\        \hline
         \multirow{4}*{$x+\tr(\alpha x+x^{2^i+1})$} & \multirow{2}*{1} & \multirow{2}*{$\bigcirc$}
         &  $\tr(\alpha+1)=0$, $d=\gcd(n,2i),$
          & \multirow{4}*{Theorem \ref{thmfinal}}\\
                                           &                  &                           &   $c \in \F_{2^d}\setminus \F_2$
                                            & \\    \cline{2-4}
                                           & \multirow{2}*{2} & \multirow{2}*{$\bigcirc$} & $\tr(\alpha+1)=0$, $d=\gcd(n,2i)$,  &        \\
                                           &                  &                           &  $c \in \F_{2^n}\setminus \F_{2^d}$    &  \\  \hline

        \multirow{2}*{$(x+\tr(\alpha x+ x^{2^i+1}))^{2^i+1}$} &  \multirow{2}*{2} &  \multirow{2}*{$\bigcirc$} &  $\tr(\alpha+1)=0$, $d=\gcd(n,i),$
         &  \multirow{2}*{Theorem \ref{maintheorem22}}\\
         &       &  &  $\frac{n}{d}>1$ is odd, $c \in \F_{2^d}\setminus \F_2$  & \\        \hline
    \end{tabular}

    \medskip
 $-$ $\tr(\alpha+1)=0$ means either that $\tr(\alpha)=1$ and $n$ is odd or that $\tr(\alpha)=0$ and $n$ is
 even.
\end{table}

\normalsize

\begin{table}[h]
    \centering
    \caption{A list of  (A)PcN functions over general finite fields, \\
    which are constructed from other (A)PcN fuctions}
    \renewcommand{\arraystretch}{1.3}
      \footnotesize
    \begin{tabular}{|c|c|c|}
        \hline
        $F(x)$ &  $_c\Delta_F$  & References \\
        \hline\hline
        $b\phi(x) + \tr^{q^n}_{q}(g(x^q+x)) $ &\multirow{2}*{2}&\multirow{2}*{\cite{BC21}}\\
        $b\phi(x) + \nr^{q^n}_{q}(g(x^q+x)) $ & & \\
        \hline
        $L(x)(\sum^{l-1}_{i=1}L(x)^{(q-1)i/l} +u)$ & $\leq 2$ & \cite{WLZ21} \\
        \hline
        $f(x)(\tr^{q}_1(x)+1) + f(x+\gm)\tr^q_1(x)$ & $\leq 2$  & \cite{WLZ21} \\
        \hline
        $L(x) + L(\gm)\tr^{q^n}_q(x)^{q-1}$ & 1  & \cite{WLZ21}\\
        \hline
        $u\phi(x) + g(\tr^{q^n}_q(x))^q +  g(\tr^{q^n}_q(x))$ & 1  & \cite{WLZ21}\\
        \hline
        $u(x^q+x) + g(\tr^{q^n}_q(x))$ & 1  & \cite{WLZ21}\\
        \hline
        $u\phi(x) +  g(\tr^{q^n}_q(x))$ & 1  & \cite{WLZ21}\\
        \hline
        $u\phi(x) + \sum^{t}_{i=1}g(x^q+x)^{(q^n-1)/d_i}$ & 2  & \cite{WLZ21}\\
        \hline
        $f(x) + u\tr(vf(x))$ & 1  & \cite{LRS22}\\
        \hline
        $L_1(x) + L_1(\gm)\tr(L_2(x))$ & 1  & \cite{LRS22}\\
        \hline
        $L(x) + \prod^{s}_{i=1}\left(\tr(x^{2^{k_i}+1}+\delta_i)\right)^{g_i} $ & $\leq 2$  & \cite{LRS22}\\
        \hline
        $L(x) + \prod^{s}_{i=1}\left(\alpha_i\tr^{q^n}_{q^m}(x^{2^{k_i}+1}+\delta_i)\right)^{g_i} $ & $\leq 2$  & \cite{LRS22}\\
        \hline
        $L(x) + u\sum^{t}_{i=1}\left(\tr^{q^n}_{q^m}(x)^{k_i}+\delta_i)\right)^{s_i} $ & 1  & \cite{LRS22}\\
        \hline
    \end{tabular}

\end{table}
\normalsize

\section{Preliminaries}
Let $\Fbn$ denote the finite field with $2^n$ elements  and let
$\Fbn^\times$ denote the multiplicative subgroup of $\Fbn$. For
positive integers $d,m,n$ such that $n =dm$, we let $\tr^{n}_{d} :
\Fbn\longrightarrow\F_{2^d}$ and  $\nr^{n}_{d} :
\Fbn\longrightarrow\F_{2^d}$ denote the relative field trace and
norm,
\begin{equation*}
\begin{split}
\tr^{n}_{d}(x) &= \sum_{k=0}^{m-1}x^{2^{dk}} = x + x^{2^d} + x^{2^{2d}} + \cdots + x^{2^{(m-1)d}},  \\
\nr^{n}_{d}(x) &= \prod_{k=0}^{m-1}x^{2^{dk}} = x^{1 + 2^d +2^{2d}
+ \cdots + {2^{(m-1)d}}}.
\end{split}
\end{equation*} When $d=1$ and $n$ is clear from the context, we shall simply write $\tr$ (resp. $\nr$) instead of $\tr^n_d$ (resp. $\nr^n_d$).
Also, when $i=d\ell$ such that $\gcd(n,i)=d$ (i.e., when
$\gcd(m,\ell)=1$), then it is easy to verify
\begin{align}
\tr^{n}_{d}(x) &= \sum_{k=0}^{m-1}x^{2^{ik}} = x + x^{2^i} + x^{2^{2i}} + \cdots + x^{2^{(m-1)i}}, \label{newtr}   \\
\nr^{n}_{d}(x) &= \prod_{k=0}^{m-1}x^{2^{ik}} = x^{1 + 2^i +2^{2i}
+ \cdots + {2^{(m-1)i}}}, \label{newnr}
\end{align}
because there is unique $0\leq j < m$ with $j\equiv \ell \pmod{m}$
such that $2^i\equiv 2^{dj} \pmod{2^n-1}$.

Two functions $F, F': \Fbnn$ are said to be \textit{extended
affine equivalent}  if there exist affine permutations $A_1,A_2 :
\Fbnn$ and an affine function $A : \Fbnn$ satisfying  $F' = A_1
\circ F \circ A_2  + A$. Also, $F$ and $F'$ are said to be
\textit{CCZ-equivalent} (introduced in \cite{CCZ98}) if there
exists an affine permutation $\mathscr{L}$ on $\Fbn\times \Fbn$
such that $ \mathscr{L}(G_F)= G_{F'}$ where $G_F$ is the
\textit{graph of the function }$F$, i.e., $G_F = \{(x,F(x)) :  x
\in \Fbn\} \subset \Fbn\times \Fbn$. Let $\#S$ denote the
cardinality of the set $S$.

\begin{definition}(\cite{EFR+20}) Let $F : \Fbnn$ be a function and $c\in\Fbn$. For $a,b \in \Fbn$,
we define $ _c D_a F(x) := F(x+a) + cF(x) $ and $ _c\Delta_F(a,b)
 := \#\{x \in \Fbn:\  _c D_a F(x)= b \}.$
    We call the quantity $$_c\Delta_F := \max \# \{_c\Delta_F(a,b)  : a,b \in \Fbn,\text{ and } a \neq 0 \text{ if } c= 1\}$$ the
     \textit{$c$-differential uniformity} of $F$. A function $F$ is called a perfect $c$-nonlinear(PcN) function and
     an almost perfect $c$-nonlinear(APcN) function if $_c \Delta_F =1$ and  $_c \Delta_F =2$, respectively.
\end{definition}

\begin{remark}
Recently, in an independent work, Bartoli and Timpanella
\cite{BT20} gave a generalization of planar functions as follows:
For $\beta \in \Fbn\setminus \F_2$,  a function $F : \Fbnn$ is
called a $\beta$-planar function in $\Fbn$ if $F(x+\gamma)+\beta
F(x)$ is a permutation on $\Fbn$ for all $\gamma \in \Fbn$.
Therefore, $\beta$-planar functions are P$_\beta$N functions.
\end{remark}

We introduce a useful lemma.
\begin{prop}(\cite{Williams75})\label{qubicequationprop}
    Let $f(x) = x^3 + ax + b$ where $a,b \in \Fbn$, $b \neq0$ with even $n$. Then $f(x)$ is irreducible over $\Fbn$ if and only if
     $\tr(a^3/b^2) = 0$ and none of the roots of $x^2+bx+a^3 = 0$ is a cube in $\Fbn$.
\end{prop}

\section{A class of non-monomial APcN permutations on $\Fbn$}
\label{sec3}
 Throughout this section, we are mainly interested in
the case where $n$ is even and $\gcd(2i,n)=\gcd(i,n)<n$.
Equivalently, we let $n,d,m,i,\ell$ be positive integers such that
  \begin{equation}\label{field setting}
 n = dm, \,\,\,  i = d\ell, \,\,\, \gcd(m,2\ell) = 1, \,\,\,
   n>d, \,\,\, n=\textrm{even}.
  \end{equation}
We will prove Theorem \ref{maintheorem} using the above
assumption. However, it should be mentioned that some of the
lemmas and propositions that we will discuss later may not need
such strong assumption, and we will explicitly state the necessary
conditions in each case.
Let $G$ be a function defined as
\begin{equation}\label{main function2}
G(x) = \begin{cases}
x^{2^i+1} &\text{ if }  x \in T\\
(x+1)^{2^i+1} &\text{ if }  x \notin T,
\end{cases}
\end{equation}
where $T = \{x \in \Fbn : \tr(x^{2^i+1}) = 0\}$.  The function $G$
can also be written as
\begin{align}
    G(x) = (x+\tr(x^{2^i+1}))^{2^i+1} = x^{2^i+1} + (x^{2^i} + x +
    1)\tr(x^{2^i+1}). \label{Gexpress}
\end{align}
 From the above form, we can verify that $G$ is a cubic function. Indeed it holds that
    $G(x) = x^{2^i+1} + (x^{2^i} + x + 1)\Sum_{j=0}^{n-1}x^{2^{i+j}+2^j},$
    which has the term $x^{2^{i+1}+2^i+2}$.
We will prove that, under the conditions in \eqref{field setting},
$G$ is an APcN permutation on $\Fbn$ for $c\in \F_{2^d}\setminus
\F_2$.

\begin{remark}
    $G(x)$ was already introduced by Budaghyan et al. \cite{BCP06} to
    derive the first example of
EA-inequivalent but CCZ-equivalent function to $x^{2^i+1}$.
However, $c$-differential uniformity is not preserved under
EA-equivalence.
    While we are assuming that $d$ is even in Theorem \ref{maintheorem},
    when $d=1$ $($i.e., $ \gcd(n,i)= 1$$)$,  the result in \cite{BCP06} implies that
    $G$ is an APN function.
\end{remark}

\begin{remark}
    When $n=d=\textrm{even}$ $($i.e., $ m= 1$$)$, then since $n | i=d\ell$, it holds that $2^i \equiv 1 \pmod{2^n-1}$ and
     $x^{2^i+1} = x^1 \cdot x^1 = x^2 $ over $\Fbn$. Then one has
     $G(x) = (x+\tr(x^{2^i+1})) ^{2^i+1} = (x+\tr(x^2))^2 = x^2 + \tr(x)$ and the result of Wu et al. \cite[Theorem     4]{WLZ21} implies that
       $G$ is PcN for all $c\in \Fbn \setminus \F_2$. Indeed if we choose $f(x) = x^2$ and
       $\gm = 1$ in \cite[Theorem 4]{WLZ21}, then $F(x) = f(x)(\tr(x) + 1) + f(x+1)\tr(x) = x^2(\tr(x) +1) + (x+1)^2\tr(x)
        = x^2+\tr(x) = G(x)$ is PcN  where the fact $f(x) = x^2$ is PcN for all $c\in \Fbn \setminus \F_2$ and $\tr(\gm) = \tr(1) = 0$ is used.
\end{remark}
 Now we introduce the following theorem.

\begin{theorem}\label{maintheorem}
    Suppose that $n$ is even with $d=\gcd(2i,n)=\gcd(i,n)<n$.  Let $G$ be the function defined in $\eqref{Gexpress}$
    and let $c\in \F_{2^{d}}\setminus\F_2$.  Then $G$ is an \text{APcN} permutation on $\Fbn$.
\end{theorem}

\noindent It will turn out that, without much difficulty,  the
$c$-differential uniformity of $G$ is bounded above by two and we
show $ {}_c \Delta_G \leq 2$ in Section \ref{sec31}. Proving that
$G$ is {\em not} PcN takes more effort and we show $ {}_c \Delta_G
= 2$ in Section \ref{sec32}.

\subsection{The proof of $ {}_c \Delta_G \leq 2$}\label{sec31}

We present some useful lemmas which will be used in both Section
\ref{sec3} and Section \ref{sec4}.

\begin{lemma}\label{2i+1lemma} The followings hold:
    \begin{enumerate}[(i)]
        \item  If $\gcd(2i,n)=\gcd(i,n)$, then $ g(x)= x^{2^i+1}$ is a permutation on $\Fbn$.
        \item   If $n$ is even, then $ h(x) = x+\tr(x^{2^i+1})$ is an involution on
        $\Fbn$. (i.e., $h\circ h(x)=x$ for all $x\in \Fbn$.)
        \item  Let $ \cl(x) = x^{2^i} + x^{2^{-i}}$ be a linear transformation on
        $\Fbn$. If $\gcd(2i,n)=d$,
         then $\ker \cl = \F_{2^{d}}$.
        \item  Let $\gcd(2i,n)=\gcd(i,n)=d$ with $n=dm$ and $i=d\ell$. Then, for all $x \in \Fbn$, one has
        $$x^{\frac{2}{2^i+1}} = x^{\frac{2^{im}+1}{2^i+1}}
        = x^{2^{(m-1)i}-2^{(m-2)i}+ \cdots - 2^i + 1}.$$
    \end{enumerate}
\end{lemma}
\begin{proof}
    For $(i)$, one has
     $$ \gcd(2^i+1,2^n-1) = \frac{2^{\gcd(2i,n)} - 1}{2^{\gcd(i,n)} - 1}
     = \frac{2^d-1}{2^d-1} = 1, $$
     where  $d=\gcd(2i,n)=\gcd(i,n)$.
       For $(ii)$, by expanding $h(h(x))$, we have
    \begin{equation}\label{lemmaeqn1}
    \begin{split}
    h(h(x)) &=  x+\tr(x^{2^i+1}) + \tr\left ( (x+\tr(x^{2^i+1}))^{2^i+1}\right ) \\
    &= x+\tr(x^{2^i+1}) + \tr\left ( x^{2^i+1} + x^{2^i}\tr(x^{2^i+1}) + x\tr(x^{2^i+1})^{2^i} +\tr(x^{2^i+1})^{2^i+1}  \right )\\
    &=x
    \end{split}
    \end{equation}
 because  $\tr\left ( x^{2^i}\tr(x^{2^i+1}) + x\tr(x^{2^i+1})^{2^i}\right
) = \tr\left ( x^{2^i}\tr(x^{2^i+1})\right ) +  \tr\left
(x^{2^i}\tr(x^{2^i+1})\right ) =  0$ and (since $n$ is even)
$\tr\left (\tr(x^{2^i+1})^{2^i+1}  \right ) =
\tr(x^{2^i+1})^{2^i+1} \tr\left ( 1 \right ) = 0 $.
     For $(iii)$, the equation $x^{2^i} = x^{2^{-i}}$ is equivalent to  $x^{2^{2i}} = x$
     so that $\ker \cl = \{x \in \Fbn : \cl(x) = 0 \} = \{x \in \Fbn : x^{2^{2i}} = x \} = \F_{2^{\gcd(2i,n)}} = \F_{2^d}$.
     Finally, for $(iv)$, it holds that $m$ is odd because $1=\gcd(2\ell, m)$ with $d=\gcd(2i,n)=\gcd(i,n)$.
     Therefore one has
    $$(2^i+1)(2^{(m-1)i}-2^{(m-2)i}+ \cdots - 2^i + 1) =2^{mi} + 1 = 2^{dm\ell} + 1 = 2^{n\ell} + 1, $$
   which implies that $x^{(2^i+1)(2^{(m-1)i}-2^{(m-2)i}+ \cdots - 2^i + 1)} = x^{2^{n\ell} + 1} = x^2$.
    Since the $\frac{1}{2^i+1}$-th power is well-defined by the claim
    $(i)$, we are done.
\end{proof}

\begin{corollary}
  If $n$ is even with $\gcd(2i,n)=\gcd(i,n)$, then $G$ is a permutation on $\Fbn$.
\end{corollary}
\begin{proof}
    $G(x) = (x+\tr(x^{2^i+1}))^{2^i+1}=g(h(x))$, where both $g(x)=
    x^{2^i+1}$ and $h(x)=x+\tr(x^{2^i+1})$ are permutations on
    $\Fbn$ by Lemma \ref{2i+1lemma}.
\end{proof}
\begin{lemma} \label{Tclosinglemma} Let $a, x \in \Fbn$.
Then one has $\tr(a^{2^i}x+ax^{2^i}) = \tr(\cl(a)x) =
\tr(\cl(x)a)$. Moreover,
    \begin{enumerate}[(i)]
        \item Let $a \in T$. Then  $x+a \in T$ if and only if $\tr(\cl(a)x)=\tr(x^{2^i+1})$.
        \item Let $a \notin T$. Then $x+a \in T$ if and only if $\tr(\cl(a)x)=\tr(x^{2^i+1})+1$.
    \end{enumerate}
\end{lemma}
\begin{proof} One has
\begin{equation*}\begin{split}
        \tr(a^{2^i}x+ax^{2^i}) = \tr(a^{2^i}x)+\tr(ax^{2^i}) &= \tr(a^{2^i}x)+\tr(a^{2^{-i}}x) =\tr(\cl(a)x)\\
        &= \tr(ax^{2^{-i}})+\tr(ax^{2^i}) =\tr(\cl(x)a).
    \end{split}
\end{equation*}
    By expanding $(x+a)^{2^i+1}$, we have
\begin{equation*}
\begin{split}
 \tr\left ((x+a)^{2^i+1}\right ) &=\tr(x^{2^i+1})+\tr(a^{2^i+1})+\tr(a^{2^i}x) + \tr(ax^{2^i}) \\
 &=\tr(x^{2^i+1})+\tr(a^{2^i+1})+\tr(\cl(a)x).
\end{split}
\end{equation*}
Assuming that $a \in T$, we can obtain $\tr((x+a)^{2^i+1})
=\tr(x^{2^i+1})+\tr(\cl(a)x)$. Therefore $x+a \in T$ if and only
if $0 = \tr((x+a)^{2^i+1}) =\tr(x^{2^i+1})+\tr(\cl(a)x)$, so we
prove the first claim.  The proof of the second claim follows
similarly.
\end{proof}

\begin{prop}\cite{MRS+21}\label{onesolutionprop}
    Let $c \in \F_{2^d} \setminus \F_2$ with $d=\gcd(2i,n)=\gcd(i,n)<n$, and let $a, b \in \Fbn$. For $\epsilon_1, \epsilon_2 \in \F_2$,
    the equation $$ (x+a+\bit_1)^{2^i+1} + c(x+\bit_2)^{2^i+1}   = b$$ has a unique solution
    $$x = \frac{a + \bit_1+c\bit_2}{c+1} +
    \left( \frac{c}{(c+1)^2}\left(  a^{2^i+1} +(a^{2^i} + a+ 1)(\bit_1+\bit_2) \right) + \frac{b}{c+1}\right)^{\frac{1}{2^i+1}}. $$
\end{prop}
\begin{proof}
    This is in fact Theorem 4 of \cite{MRS+21}. However since we actively use the equation \eqref{pcnmono} later,
    we provide a proof for better understanding. By expanding
    left-hand side of the given equation, one has
    \begin{equation}\label{a1a2a3eqn}
    (x+a+\bit_1)^{2^i+1} + c(x+\bit_2)^{2^i+1} = (c+1)x^{2^i+1} + a_1x^{2^i} +a_2x + a_3
    \end{equation}
    where $a_1 = a + \bit_1 + c\bit_2, a_2 = a^{2^i} + \bit_1 + c\bit_2 $ and $a_3= a^{2^i+1} +\bit_1a^{2^i} + \bit_1a + \bit_1 + c\bit_2$.
    Since $c \in \F_{2^d} \subset \F_{2^i} $, one has  $c^{2^i} = c, a_1^{2^i} = a_2$ and the right-hand side of \eqref{a1a2a3eqn} becomes
    \begin{equation*}
    \begin{split}
    (c+1)x^{2^i+1} + a_1x^{2^i} +a_2x + a_3  &= (c+1)x^{2^i+1} + a_1x^{2^i} +a_1^{2^i}x + a_3  \\
    &= (c+1)\left (x + \dfrac{a_1}{c+1}\right)^{2^i+1}  + \dfrac{a_1^{2^i+1}}{c+1} + a_3.
    \end{split}
    \end{equation*} Therefore the original equation $(x+a+\bit_1)^{2^i+1} + c(x+\bit_2)^{2^i+1}   = b$ can be written as
    \begin{equation}\label{pcnmono}
      \left(x+\frac{a + \bit_1 + c\bit_2}{c+1}\right)^{2^i+1}
    =  \frac{c}{(c+1)^2}\left( a^{2^i+1} + (a^{2^i} + a+ 1)(\bit_1+\bit_2) \right) + \frac{b}{c+1}.
    \end{equation}
\end{proof}

To compute $_c \Delta_G(a,b)$, we divide the set  $\{x \in \Fbn :
\ _cD_aG(x) = b\}$ into four disjoint subsets as follow:
\begin{equation}\label{Diff-partition-eqn}
    \begin{split}
        \{x \in \Fbn : \ _cD_aG(x) = b\}=&\{x \in T :\ _cD_aG(x) = b\} \cup \{x \in \Fbn \setminus T :\ _cD_aG(x) = b\} \\
        =  &\bigcup_{ \bit \in \F_2}\{x \in T :\ _cD_aG(x) = b,\ \tr(\cl(a)x) = \bit \} \\
        &\cup \bigcup_{ \bit \in \F_2} \{x \in \Fbn \setminus T :\ _cD_aG(x) = b,\ \tr(\cl(a)x) = \bit \}.
    \end{split}
\end{equation}

Therefore it holds that
\begin{equation} \label{cDeltasumeq}
\begin{split}
_c \Delta_G(a,b) =  &\Sum_{ \bit \in \F_2} \# \{x \in T : G(x+a) + cG(x)= b,\ \tr(\cl(a)x) = \bit \}  \\
&+ \Sum_{ \bit \in \F_2} \# \{x \in \Fbn\setminus T :G(x+a) +
cG(x) = b,\ \tr(\cl(a)x) = \bit \}.
\end{split}
\end{equation}

From the construction of $G$ and by Lemma \ref{Tclosinglemma}, the Equation \eqref{cDeltasumeq} implies the following:\\

$\bullet$ For $a \in T$, it holds that
     \begin{equation}\label{ainTeqn}
     \begin{split}
     _c \Delta_G(a,b) =   &\# \{x \in T :  (x+a)^{2^i+1} + c x^{2^i+1}   = b,\ \tr(\cl(a)x) = 0 \}  \\
     &+   \# \{x \in T : (x+a+1)^{2^i+1} + cx^{2^i+1}    = b,\ \tr(\cl(a)x) = 1 \}  \\
     &+   \# \{ x \in \Fbn\setminus T :    (x+a+1)^{2^i+1}  + c (x+1)^{2^i+1}= b,\ \tr(\cl(a)x) = 0 \} \\
     &+   \# \{x \in \Fbn\setminus T :  (x+a)^{2^i+1} + c (x+1)^{2^i+1} = b,\ \tr(\cl(a)x) = 1 \}. \\
     \end{split}
     \end{equation}

     $\bullet$ For $a \notin T$, it holds that
     \begin{equation}\label{anotinTeqn}
     \begin{split}
     _c \Delta_G(a,b) =  &\# \{x \in T :  (x+a+1)^{2^i+1} + cx^{2^i+1}   = b,\ \tr(\cl(a)x) = 0 \}   \\
     &+   \# \{x \in T : (x+a)^{2^i+1} + cx^{2^i+1}    = b,\ \tr(\cl(a)x) = 1 \}  \\
     &+   \# \{ x \in \Fbn\setminus T :  (x+a)^{2^i+1} + c(x+1)^{2^i+1}   = b,\ \tr(\cl(a)x) = 0 \}  \\
     &+   \# \{x \in \Fbn\setminus T :  (x+a+1)^{2^i+1} + c(x+1)^{2^i+1}  = b,\ \tr(\cl(a)x) = 1 \}.
     \end{split}
     \end{equation}
Let the sets $S_{\bit_1\bit_2}(a,b)$ and $S'_{\bit_1\bit_2}(a,b)$ be defined as
\begin{equation}\label{Sabdefine}
\begin{split}
S_{\bit_1\bit_2}(a,b) := \{x \in \Fbn : (x+a+\bit_1)^{2^i+1} &+ c(x+\bit_2)^{2^i+1} = b, \\
&\tr(\cl(a)x) =\bit_1 +\bit_2 , \ \tr(x^{2^i+1}) = \bit_2\},\\
S'_{\bit_1\bit_2}(a,b) := \{x \in \Fbn : (x+a+\bit_1)^{2^i+1} &+ c(x+\bit_2)^{2^i+1}  = b, \\
&\tr(\cl(a)x) =\bit_1 +\bit_2+1 , \ \tr(x^{2^i+1}) = \bit_2\},
\end{split}
\end{equation}
where $a , b \in \Fbn,  c \in \F_{2^d}\setminus \F_2$ and
$\epsilon_1, \epsilon_2 \in \F_2$.
 Then  \eqref{ainTeqn} and  \eqref{anotinTeqn} can be simply written as
\begin{equation}\label{4partpropre} 
 _c \Delta_G(a,b) =
 \begin{cases}
  \,\, \Sum_{\bit_1, \bit_2 \in \F_2} \#S_{\bit_1\bit_2}(a,b) &\text{for } a \in T  \\
   \,\, \Sum_{\bit_1, \bit_2 \in \F_2} \#S'_{\bit_1\bit_2}(a,b) &\text{for } a \notin
   T .
\end{cases}  \\
\end{equation}

\noindent By Proposition \ref{onesolutionprop},  it holds that
$\#S_{\bit_1\bit_2}(a,b) \leq 1 \text{ for } a \in T$ and
$\#S'_{\bit_1\bit_2}(a,b) \leq 1 \text{ for } a \notin T$ so that
one has $_c \Delta_G(a,b) \leq 4$ for all $ a, b\in \Fbn$. By
further assuming that $n$ is even, the following lemma implies
that $\Sum_{\bit_1, \bit_2 \in \F_2} \#S_{\bit_1\bit_2}(a,b) \leq
2$ (resp., $\Sum_{\bit_1, \bit_2 \in \F_2}
\#S'_{\bit_1\bit_2}(a,b) \leq 2$) when $a\in T$ (resp., when $a
\notin T$) so that one has $_c \Delta_G \leq 2$.


\begin{lemma}\label{threeimpossiblelemma}
    Let $c \in \F_{2^d} \setminus \F_2$ with $d=\gcd(2i,n)=\gcd(i,n)<n$ and let $n$ be even.
     Let $a, b \in \Fbn$ and
     let $S_{\bit_1\bit_2}(a,b)$, $S'_{\bit_1\bit_2}(a,b)$ be
    defined as in \eqref{Sabdefine}. Then the followings are satisfied:
    \begin{enumerate}[(i)]
        \item At least one of $S_{00}(a,b)$ and $S_{11}(a,b)$ is the empty set.
        \item At least one of  $S_{01}(a,b)$ and $S_{10}(a,b)$ is the empty set.
        \item At least one of $S_{00}'(a,b)$ and $S_{11}'(a,b)$ is the empty set.
        \item At least one of  $S_{01}'(a,b)$ and $S_{10}'(a,b)$ is the empty set.
    \end{enumerate}
\end{lemma}
\begin{proof}   We only prove the first claim, and the rest of the claims can be proved in a similar manner.
It is obvious that $x_0$ is a solution of $ (x+a)^{2^i+1}  +
cx^{2^i+1}  = b$ if and only if $x_0+1$ is
 a solution of $(x+a+1)^{2^i+1} + c (x+1)^{2^i+1}    = b$.  By Proposition \ref{onesolutionprop},
 the both equations have unique solutions, and
   $x_0 \in S_{00}(a,b)$ implies that  $x_0+1$ is a unique
  solution of
  $ (x+a+1)^{2^i+1} + c(x+1)^{2^i+1}   = b$.  Since $n$ is even,
 it holds that $x_0 \in T$ if and only if $x_0+1 \in T$, which means that $S_{11}(a,b) = \phi$  and vice versa.
\end{proof}

\subsection{The proof of $ {}_c \Delta_G =2$}\label{sec32}

\noindent Lemma \ref{threeimpossiblelemma} combined with the
equation \eqref{4partpropre} shows that
 $_c \Delta_G \leq 2$ for any $c\in \F_{2^d}\setminus \F_2$. To
 show that $G$ is APcN (i.e., $_c \Delta_G = 2$), we need to prove that,
 for a given $c$, there exist $a,b \in \Fbn$ such that $_c \Delta_G(a,b) =
 2$.

   \begin{lemma}\label{independetlemma}
        Suppose that $n=dm$ is even with $d=\gcd(2i,n)=\gcd(i,n)<n$.  Let $c \in \F_{2^d} \setminus
        \F_2$ and $w \in \Fbn^\times$.
   Considering $\Fbn$ as an $m$-dimensional vector space over $\F_{2^d}$, one has the
    followings.
    \begin{enumerate}[(i)]
        \item If $m >3$ with $3 \nmid m$, then $\left\{ \frac{1}{c} , \frac{1}{w^{2^i+1}}, w+w^{2^i}\right\}$ is linearly independent
        over $\F_{2^{d}}$ for $w \in \Fbn \setminus \F_{2^d}$.
        \item If $m >3$ with $3 \mid m$, then $\left\{ \frac{1}{c} , \frac{1}{w^{2^i+1}}, w+w^{2^i}\right\}$ is linearly independent
        over $\F_{2^{d}}$ for $w \in \Fbn\setminus \F_{2^{3d}}$.
        \item If $m =3$,  then  $\left\{ \frac{1}{c} , \frac{1}{w^{2^i+1}}, w+w^{2^i}\right\}$ is linearly dependent over $\F_{2^{d}}$.
         More precisely,
          $$w+w^{2^i} = c\tr^n_d(w)\cdot \frac{1}{c} + \nr^n_d(w)\cdot \frac{1}{w^{2^i+1}}. $$
    \end{enumerate}
   \end{lemma}
   \begin{proof}
    For $(i)$ and $(ii)$, suppose that
    \begin{equation}\label{eqn6}
    \dfrac{a_1}{c} +\dfrac{a_2}{w^{2^i+1}} + a_3 (w+w^{2^i}) = 0
    \end{equation}
      for some $a_1,a_2,a_3 \in \F_{2^d}$.
    Letting $A = \dfrac{a_2}{w^{2^i+1}} + a_3 (w+w^{2^i})  = \dfrac{a_1}{c} \in \F_{2^d} \subset \F_{2^i}$, one has
    \begin{equation}\label{eqn7}
    \begin{split}
    0 = A^{2^i} + A  &= \dfrac{a_2}{w^{2^{2i}+2^{i}}} + a_3 (w^{2^{i}}+w^{2^{2i}})  + \dfrac{a_2}{w^{2^i+1}} + a_3 (w+w^{2^i}) \\
    &= \dfrac{a_2+a_2w^{2^{2i}-1}}{w^{2^{2i}+2^{i}}} + a_3 (w+w^{2^{2i}})   \\
    &=\dfrac{w^{2^{2i}-1} + 1}{w^{2^{2i}+2^{i}}}\left(a_2 + a_3w^{2^{2i} + 2^i + 1}\right).  \\
    \end{split}
    \end{equation}
    Since $w \in \Fbn \setminus \F_{2^d}$ and $\gcd(n,2i) = d$, it holds that  $w^{2^{2i}-1} + 1 \neq 0$,
     which implies that $a_2 + a_3w^{2^{2i} + 2^i + 1} = 0$.
    When $m > 3$ and $3 \nmid m$, one gets $(w^{2^{2i} + 2^i + 1})^{2^i-1}= w^{2^{3i}-1} \neq 1$
     and so that $w^{2^{2i} + 2^i + 1} \notin \F_{2^d}^{\times}$.
     Therefore $a_2 + a_3w^{2^{2i} + 2^i + 1} = 0$ implies that
       $a_3 = a_2 = 0$ and consequently $a_1 = 0$, which proves
       $(i)$. For the assertion $(ii)$, using
         $3 \mid m >3 $ and $ w \in \Fbn\setminus \F_{2^{3d}}$,  one gets $w^{2^{3i}-1} \neq 1$ and $w^{2^{2i} + 2^i + 1} \notin \F_{2^d}^{\times}$
          in a similar manner, and therefore  $a_3 = a_2 =a_1 = 0$.

    For $(iii)$, suppose that $m=3$. By the equations \eqref{newtr} and \eqref{newnr}, one has
     $\tr^n_d(x) = x+x^{2^i}+x^{2^{2i}}$ and  $\nr^n_d(x) = x^{1+2^i+2^{2i}}$.
     Therefore it follows that
    \begin{equation*}
    \begin{split}
    w+w^{2^i} = w+w^{2^i}+w^{2^{2i}} +w^{2^{2i}} &= w+w^{2^i}+w^{2^{2i}} + \frac{w^{1+2^i+2^{2i}}}{ w^{1+2^i}}\\
    &= \tr^n_d(w) + \frac{\nr^n_d(w)}{ w^{1+2^i}} \\
    &= c\tr^n_d(w)\cdot \frac{1}{c} + \nr^n_d(w)\cdot
    \frac{1}{w^{2^i+1}},
    \end{split}
    \end{equation*}
 where the coefficients $c\tr^n_d(w), \nr^n_d(w)$ are  in
 $\F_{2^{d}}$.
   \end{proof}
\begin{lemma}\label{trnrlemma}
    Let $m=3$ and $d$ be even such that $n=3d$. Let $H = \{x^3  : x \in  \F_{2^{d}}^{\times}\}.$ Then the followings hold:
    \begin{enumerate}[$(i)$]
        \item For any $\zeta \in \F_{2^{d}}^\times \setminus H,$
         one has  $\F_{2^{d}}^\times = \zeta H + \zeta^2 H \overset{\rm{def}}{=} \{\zeta h_1+\zeta^2 h_2 : h_1, h_2 \in H \}$.
        \item For any $c' \in \F_{2^{d}} \setminus \F_2$, there exists $w \in \Fbn\setminus\F_{2^{d}}$
        such that $\tr^n_d(w) = 0$ and $\nr^n_d(w) = c'$.
    \end{enumerate}
\end{lemma}
\begin{proof}
     Claim $(i)$ seems to be a folklore result and we provide a proof in Appendix \ref{appendix}.
     Now, by the claim $(i)$, for a given $c' \in \F_{2^{d}} \setminus \F_2$, it holds that $c' \in \zeta H + \zeta^2 H$
      where $\zeta \in \F_{2^{d}} \setminus H$. There exist $h_1,h_2 \in H$ such that $c' = \zeta h_1 + \zeta^2 h_2$
      where  $\zeta h_1$ and $ \zeta^2 h_2$ are not cubes in $\F_{2^d}$. Since $\zeta^3h_1h_2$ is a cube,
       there exists $\widetilde{c} \in \F_{2^{d}}$ such that $ \widetilde{c}^3 = \zeta^3h_1h_2$.
    Then we have the following factorization in $\F_{2^d}[x]$
     $$x^2 + c'x + \widetilde{c}^3 = (x+\zeta h_1) (x+ \zeta^2 h_2),$$
     which implies that the cubic $x^3+\widetilde{c}x +c'$ is irreducible over $\F_{2^d}$ by Proposition \ref{qubicequationprop}.
     Since $m=3$, there exists $w \in \Fbn \setminus \F_{2^{d}}$ such that $x^3+\widetilde{c}x +c' = (x+w)(x+w^{2^i})(x+w^{2^{2i}})
      = x^3 + \tr^n_d(w)x^2 + \widetilde{c}x + \nr^n_d(w)$, which completes the proof of the claim $(ii)$.

\end{proof}

\begin{lemma}\label{3trlemma} Let $c \in \F_{2^d} \setminus \F_2$ with $d=\gcd(2i,n)=\gcd(i,n)<n$ and let $n=dm$ be even.
Then, one can choose $w,z \in \Fbn$
  satisfying all of the following trace conditions:\\
$(i)$ $\tr^n_d\left(\dfrac{z}{c}\right) = 1$,  \qquad $(ii)$
$\tr^n_d \left(\dfrac{z}{w^{2^i+1}}\right)
 = 1$, \qquad $(iii)$  $\tr\left(\dfrac{z(w+w^{2^i})}{c(c+1)}\right) = 1$.
\end{lemma}
\begin{proof}
Let $\beta_1 = \frac{1}{c}$, $\beta_2 = \frac{1}{w^{2^i+1}}$ and
$\beta_3 = w+w^{2^i}$. We shall consider two cases: $m>3$ and
$m=3$.

\medskip

\textbf{Case 1 :} $m > 3$.   By Lemma \ref{independetlemma}, for
any $w \in \Fbn\setminus \F_{2^{3d}}$,
$\{\beta_1,\beta_2,\beta_3\}$ is a linearly independent subset
over $\F_{2^{d}}$. Then one can find $\beta_4, \cdots, \beta_m \in
\Fbn$ such that $\{\beta_1, \beta_2, \cdots, \beta_m\}$ is a basis
of $\Fbn$ over $\F_{2^{d}}$.  Also, there is a unique dual basis
$\{\gm_1, \cdots, \gm_m\}$ $($ See \cite{LN07}$)$ for $\Fbn$ over
$\F_{2^d}$ such that
$$\tr^n_d(\beta_i\gm_j) = \begin{cases}
1 &\text{ if } i = j\\
0 &\text{ if } i \neq j .
\end{cases}$$
Choosing any $c' \in \F_{2^d}$ satisfying
$\tr^d_1\left(\frac{c'}{c(c+1)}\right) = 1$ and letting $z = \gm_1
+ \gm_2 + c'\gm_3$, one gets
\begin{equation*}
\begin{split}
\tr^n_d\left(\dfrac{z}{c}\right)  &= \tr^n_d\left (\beta_1(\gm_1 +
\gm_2 + c'\gm_3)\right ) = 1, \\
\tr^n_d\left(\dfrac{z}{w^{2^i+1}}\right)  &= \tr^n_d\left (\beta_2(\gm_1 + \gm_2 + c'\gm_3)\right ) = 1, \\
\tr\left(\dfrac{z(w+w^{2^i})}{c(c+1)}\right)  &= \tr\left
(\frac{\beta_3(\gm_1 + \gm_2 + c'\gm_3)}{c(c+1)}\right ) =
\tr^d_1\left( \tr^n_d\left (\frac{\beta_3(\gm_1 + \gm_2 +
c'\gm_3)}{c(c+1)}\right )\right) =
\tr^d_1\left(\frac{c'}{c(c+1)}\right) = 1.
\end{split}
\end{equation*}

\textbf{Case 2 :} $m = 3$. Choose any $c' \in \F_{2^d}$ such that
$\tr^d_1\left(\frac{c'}{c(c+1)}\right) = 1$. Using Lemma
\ref{trnrlemma},  one can also find $w \in \Fbn \setminus
\F_{2^{d}}$ satisfying $\tr^n_d(w) = 0$ and $\nr^n_d(w) = c'$. By
Lemma \ref{independetlemma}, it holds that
$$\beta_3 =
c\tr^n_d(w)\beta_1 + \nr^n_d(w)\beta_2 = c'\beta_2.$$
 Since $w
\notin \F_{2^{d}}$, $\{\beta_1,\beta_2\}$ is linearly independent
over $\F_{2^{d}}$. Thus one can find $\beta_3'\in \Fbn$ such that
$\{\beta_1, \beta_2, \beta_3'\}$ is a basis of $\Fbn$ over
$\F_{2^{d}}$. Letting $z = \gm_1 + \gm_2$ where $\{\gm_1, \gm_2,
\gm_3\}$ is a unique dual basis of $\{\beta_1, \beta_2,
\beta_3'\}$,   one gets
\begin{equation*}
\begin{split}
\tr^n_d\left(\dfrac{z}{c}\right)  &= \tr^n_d\left (\beta_1(\gm_1 +
\gm_2 )\right ) = 1, \\
\tr^n_d\left(\dfrac{z}{w^{2^i+1}}\right)  &= \tr^n_d\left (\beta_2(\gm_1 + \gm_2)\right ) = 1, \\
\tr\left(\dfrac{z(w+w^{2^i})}{c(c+1)}\right)  &= \tr\left
(\frac{\beta_3(\gm_1 + \gm_2 )}{c(c+1)}\right )  = \tr^d_1\left(
\tr^n_d\left (\frac{c'\beta_2(\gm_1 + \gm_2)}{c(c+1)}\right
)\right) = \tr^d_1\left(\frac{c'}{c(c+1)}\right) = 1.
\end{split}
\end{equation*}
\end{proof}

Let \begin{equation}\label{faeqn}
  f_a(x) = (x+a)^{2^i+1} + cx^{2^i+1}
  \end{equation} where $c \in \F_{2^d} \setminus \F_2$ and $a \in \Fbn$.
  Recall that Proposition  \ref{onesolutionprop} implies
  that $f_a$ is bijective on $\Fbn$. Therefore, for a given $x \in \Fbn$,
  there exist unique $b$ and  $y$ in $\Fbn$ satisfying
   \begin{equation}\label{f_axbf_a+1yeqn}
   f_a(x) = b = f_{a+1}(y).
   \end{equation}
   We  show that one can choose $x,y \in \Fbn$ satisfying \eqref{f_axbf_a+1yeqn} and  the following trace condition,
    which play an important role in the proof of Theorem
    \ref{maintheorem}.

  \begin{prop}\label{Jprop}
  Let $c \in \F_{2^d} \setminus \F_2$ with $d=\gcd(2i,n)=\gcd(i,n)<n$ and let $n=dm$ be even.
  Then, there exist $x, y \in \Fbn$
   and $a \in \Fbn$ such that
     $$ f_a(x) = f_{a+1}(y) \quad \textrm{ and } \quad \tr\left((c+1)(x^{2^i+1} +y^{2^i+1})\right) =1.$$
  \end{prop}
\begin{proof}
    By Lemma \ref{3trlemma}, one can find $w,z \in \Fbn$
    satisfying $$\tr^n_d\left(\dfrac{z}{c}\right) = 1,\quad \tr^n_d \left(\dfrac{z}{w^{2^i+1}}\right)  = 1,
    \quad  \tr\left(\dfrac{z(w+w^{2^i})}{c(c+1)}\right) = 1.$$ Since $m$ is odd, it holds that
     $\tr^n_d\left(\dfrac{z}{c}+1\right) = \tr^n_d \left(\dfrac{z}{w^{2^i+1}}+1\right)  = 0$ and, by Hilbert Theorem $90$,
      there exist
      $a, X \in \Fbn \setminus \F_{2^{d}}$ satisfying
       $$\frac{z}{c} = a^{2^i} + a+ 1 \quad \textrm{  and  } \quad
      \frac{z}{w^{2^i+1}} = X^{2^i} + X + 1. $$
      Letting $\theta = wX$ and $\tau = \theta + w$, one has
    \begin{equation}\label{wzeqn}
   z= \theta^{2^i}w +\theta w^{2^i}   +w^{2^i+1} = \theta^{2^i+1} +\tau^{2^i+1}  \quad \textrm{  and  }\quad   w= \theta+ \tau.
    \end{equation}
    Letting $b = \frac{\theta^{2^i+1} + ca^{2^i+1}}{c+1}$, we get
      \begin{equation}\label{thetataueqn}\begin{split}
     \theta^{2^i+1} &= (c+1)b + ca^{2^i+1},\\
\tau^{2^i+1} &= \theta^{2^i+1}+z =  (c+1)b + ca^{2^i+1} +
c(a^{2^i} + a+ 1) =(c+1)b + c(a+1)^{2^i+1}.
     \end{split}
    \end{equation}


\noindent Therefore, letting $x = \frac{a+\theta}{c+1}$ and
$y=\frac{a+1+\tau}{c+1}$, we have the following equations
\begin{equation*}
\begin{split}
  \left(x+\frac{a}{c+1}\right)^{2^i+1} &=   \frac{\theta^{2^i+1}}{(c+1)^{2^i+1}} =   \frac{ca^{2^i+1} + (c+1)b}{(c+1)^{2}}, \\
  \left(y+\frac{a+1}{c+1}\right)^{2^i+1} &=   \frac{\tau^{2^i+1}}{(c+1)^{2^i+1}} =  \frac{ c(a+1)^{2^i+1} +
  (c+1)b}{(c+1)^{2}},
\end{split}
\end{equation*}
 which implies   that $f_a(x) = b = f_{a+1}(y)$ by Proposition \ref{onesolutionprop}.
 Finally, using the property  $\tr^n_d(u^{2^i}+u)=0$ for all $u\in
 \Fbn$ repeatedly,
\begin{equation*}
\begin{split}
  \tr&\left((c+1)(x^{2^i+1} +y^{2^i+1})\right) =   \tr^d_1\left(\tr^n_d\left(\dfrac{(a+\theta)^{2^i+1}+(a+1+\tau)^{2^i+1}}{c+1}\right)\right) \\
  &=\tr^d_1\left(\dfrac{1}{c+1}\tr^n_d\left((a+\theta)^{2^i+1}+(a+1+\tau)^{2^i+1}\right)\right)\\
    &=\tr^d_1\left(\dfrac{1}{c+1}\tr^n_d\left(a^{2^i}(\theta +\tau)+ a(\theta+\tau)^{2^i} + a^{2^i} + a + \theta^{2^i+1}
     + \tau^{2^i+1} +\tau^{2^i} +\tau +1 \right)\right)\\
     &=\tr^d_1\left(\dfrac{1}{c+1}\tr^n_d\left(a^{2^i}w+ aw^{2^i} + z  +1 \right)\right) \qquad (\because \textrm{ by }\eqref{wzeqn})\\
     &=\tr^d_1\left(\dfrac{1}{c+1}\tr^n_d\left((a^{2^i}+a)(w +w^{2^i})  + z  +1 \right)\right)\\
     &=\tr^d_1\left(\dfrac{1}{c+1}\tr^n_d\left(\dfrac{z}{c}(w +w^{2^i})  + z  +1 \right)\right)\\
     &=\tr^d_1\left(\dfrac{1}{c+1}\tr^n_d\left(\dfrac{z}{c}(w +w^{2^i})\right)  +1 \right)
       \qquad (\because  m \text{ is odd} \textrm{ and } \tr^n_d(z+1) = c+1.)\\
     &=\tr\left(\dfrac{z}{c(c+1)}(w +w^{2^i})\right) =1. \qquad (\because d \text{ is even and by Lemma } \ref{3trlemma}.)
\end{split}
\end{equation*}
\end{proof}

\begin{example} \hfill\break

    \noindent 1. Let $n = 10$, $d= 2$, $m = 5$ and $i = 2$. We choose the field   $\F_{2^n} = \F_2[x] /
     \langle  x^{10} + x^6 + x^5 + x^3 + x^2 + x + 1  \rangle   \text{ and a primitive root } \xi \in \F_{2^{n}}$,
       $c = \xi^{341} \in \F_4\setminus\F_2$ and $c' = \xi^{682}$. Note that $\tr^d_1\left(\frac{c'}{c(c+1)}\right) = \tr^d_1(\xi^{682})=1$.
       Let $w = \xi \in \Fbn \setminus \F_{2^{d}}$ and $\{\beta_1,\beta_2,\beta_3\}  = \{ 1/c ,1/w^{2^i+1}, w+w^{2^i}\}
       =\{\xi^{682}, \xi^{1018}, \xi^{746}\}$. Then $\{\beta_1,\beta_2,\beta_3\}$ can
       be extended to a basis $B$ for $\Fbn$ with its unique dual basis $B^\star$:
    \begin{equation*}
     B = \{\xi^{682}, \xi^{1018}, \xi^{746},\xi^{1},\xi^{2}\},
     \quad
     B^\star = \{\xi^{341} ,\xi^{684}, \xi^{41} ,\xi^{405} ,\xi^{349}\}.
    \end{equation*}
    By Lemma \ref{3trlemma}, for $w = \xi$ and $z = \gm_1+\gm_2+c'\gm_3 = \xi^{341} + \xi^{684}+ \xi^{682}\cdot \xi^{41} = \xi^{916}$,
    it holds that
    $$ (i) \tr^n_d\left(\dfrac{z}{c}\right) = \tr^n_d\left(\xi^{575}\right) = 1; \ (ii) \tr^n_d \left(\dfrac{z}{w^{2^i+1}}\right)
     = \tr^n_d(\xi^{911}) = 1;\ (iii)  \tr\left(\dfrac{z(w+w^{2^i})}{c(c+1)}\right) = \tr(\xi^{639})= 1. $$
    We can choose any $a \in \{\xi^{334},\xi^{357},\xi^{375},\xi^{981}\}$ and $X \in \{\xi^{345},\xi^{501},\xi^{595},\xi^{861}\}$
     satisfying $z/c = a^{2^i} + a+ 1$ and $z/w^{2^i+1} = X^{2^i} + X + 1$. We now pick a pair $(a,X) = (\xi^{334},\xi^{345})$.
      Following  the proof of Proposition \ref{Jprop}, we get $\theta = wX = \xi^{346}$ and $\tau = \theta + w = \xi^{502}$.
    Therefore we have  $$x = \frac{a+\theta}{c+1} = \xi^{586},\quad y =
    \frac{a+1+\tau}{c+1}
    = \xi^{718},\quad b = \frac{\theta^{2^i+1}+ca^{2^i+1}}{c+1} = \xi^{303},$$
    which satisfy $f_a(x) = b = f_{a+1}(y)$ and $\tr\left((c+1)(x^{2^i+1}+y^{2^i+1})\right)= \tr(\xi^{807}) = 1$. \\

     \smallskip
    \noindent  2.   Let $n = 6$, $d= 2$, $m = 3$ and $i = 2$.   We choose the field     $\F_{2^6} = \F_2[x] /
    \langle x^6 + x^4 + x^3 + x + 1 \rangle    \text{ and a primitive root } \xi \in \F_{2^6}$,  $c = \xi^{21}
    \in \F_4\setminus\F_2$ and $c' = \xi^{42}$. Then $\tr^n_d(w) = 0$ and $\nr^n_d(w) = c'$  if and only if
    $w \in \{\xi^{14},\xi^{35}, \xi^{56}\}$.  We now choose $w = \xi^{14}$ and $\{\beta_1,\beta_2\}
     = \{ 1/c ,1/w^{2^i+1}\} =\{\xi^{42}, \xi^{56}\}$. Then $\{\beta_1,\beta_2\}$ can be extended to
      a basis $B$ for $\Fbn$ with its unique dual basis $B^\star$:
    \begin{equation*}
    B = \{\xi^{42}, \xi^{56}, \xi^{2}\}, \quad
    B^\star = \{\xi^{37} ,\xi^{7}, \xi^{14} \}.
    \end{equation*}
    By Lemma \ref{3trlemma}, for  $w = \xi^{14}$ and $z = \gm_1+\gm_2 = \xi^{37} + \xi^7 = \xi^{12}$, it holds that
    $$ (i) \tr^n_d\left(\dfrac{z}{c}\right) = \tr^n_d\left(\xi^{54}\right) = 1; \ (ii) \tr^n_d \left(\dfrac{z}{w^{2^i+1}}\right)
      = \tr^n_d(\xi^{5}) = 1;\ (iii)  \tr\left(\dfrac{z(w+w^{2^i})}{c(c+1)}\right)= \tr(\xi^{47}) = 1. $$
      Similarly to the previous example, we can choose any $a \in \{\xi^{9},\xi^{27},\xi^{47},\xi^{61}\}$
      and $X \in \{\xi^{3},\xi^{13},\xi^{20},\xi^{57}\}$. We now pick a pair $(a,X) = (\xi^{9},\xi^{3})$.
       We also get $\theta = wX = \xi^{17}$ and $\tau = \theta + w = \xi^{27}$.
       Therefore we get
    then $$x = \frac{a+\theta}{c+1} = \xi^{37},\quad y = \frac{a+1+\tau}{c+1} = 0,\quad
    b = \frac{\theta^{2^i+1}+ca^{2^i+1}}{c+1} = \xi^{9},$$
    which satisfy $f_a(x) = b = f_{a+1}(y)$ and $\tr\left((c+1)(x^{2^i+1}+y^{2^i+1})\right)= \tr(\xi^{38})= 1$.\\
\end{example}
The following proposition shows that there exist $a,b \in \Fbn$
such that $_c\Delta_G(a,b) \geq 2$.
\begin{prop} \label{geq2prop}
Let $G$ be the function defined in  $\eqref{Gexpress}$.  Let $c
\in \F_{2^d} \setminus \F_2$ with $d=\gcd(2i,n)=\gcd(i,n)<n$ and
let $n$ be even.
 Then the
followings hold:
    \begin{enumerate}[$(i)$]
        \item   If $a \in \F_{2^{d}}$, then     $_c \Delta_G(a,b) = 1$ for all $b \in \Fbn$.
        \item   There exists $a\in \Fbn \setminus \F_{2^{d}} ,\ b \in \Fbn$ such that   $_c \Delta_G(a,b) \geq 2.$
    \end{enumerate}
\end{prop}

\begin{proof}
    Since $h(x) = x + \tr(x^{2^i+1})$  defined in  Lemma \ref{2i+1lemma}
     is an involution, the image set of $_c D_aG$ is exactly equal to that of $(_c D_aG) \circ h$.
     Therefore, to determine whether $G$ is PcN or APcN, it suffices to consider the image of the map $(_c D_aG) \circ h $.
     Since $G(x) = h(x)^{2^i+1}$, one has
    \begin{equation}\label{eqn2}
    \begin{split}
     (_c D_aG) \circ h (x) &= G(h(x) + a) + cG(h(x))   \\
        &= \{h(h(x) + a)\}^{2^i+1} + c\{h(h(x))\}^{2^i+1}
      = \{h(h(x) + a)\}^{2^i+1} + cx^{2^i+1}.
    \end{split}
    \end{equation}
     Now, the set $\Fbn$ can be written as a disjoint union
\begin{equation}\label{eqn3}
\Fbn = V_a \cup W_a,
\end{equation}
where
\begin{equation}\label{VaWadef}
 V_a = \{x \in \Fbn : \tr(\cl(a)x) = \tr(a^{2^i+1})\}
    \, \textrm{ and } \,
 W_a = \{x \in \Fbn : \tr(\cl(a)x) = \tr(a^{2^i+1})+1\}.
\end{equation}
   Recall that $f_a(x) = (x+a)^{2^i+1} + cx^{2^i+1}$  is bijective
   for  any $a  \in \Fbn$ and  $c \in \F_{2^{d}}\setminus \F_2$.
  For $\epsilon \in \F_2$, one has
\begin{equation*}
\begin{split}
(_c D_{a+\epsilon}G \circ h) (x) = f_a(x) &\Longleftrightarrow
   \{h(h(x)+a+\epsilon)\}^{2^i+1} = (x+a)^{2^i+1}  \qquad (\text{by } \eqref{eqn2}  ) \\
&\Longleftrightarrow  h(h(x)+a+\epsilon) = x+a  \qquad (\text{by Lemma } \ref{2i+1lemma}(i)  ) \\
&\Longleftrightarrow  h(x+a) + h(x)=a+\epsilon    \qquad ( h \text{ is an involution}   ) \\
&\Longleftrightarrow \tr(\cl(a)x) = \tr(a^{2^i+1})+\epsilon.
\end{split}
\end{equation*}
In a similar way (or by replacing $a$ with $a+1$ and $\epsilon$
with $\epsilon+1$), one has
$$
(_c D_{a+\epsilon}G \circ h) (x) = f_{a+1}(x) \Longleftrightarrow
h(x+a+1) + h(x)=a+\epsilon \Longleftrightarrow \tr(\cl(a)x) =
\tr(a^{2^i+1}) + 1+\epsilon.
$$
\noindent
 Consequently, one has
\begin{equation*}
\begin{split}
V_a &= \{x \in \Fbn : (_c D_aG \circ h) (x) = f_a(x) \} = \{x \in \Fbn : (_c D_{a+1}G \circ h) (x) = f_{a+1}(x) \}, \\
W_a &= \{ x\in \Fbn : (_c D_aG \circ h) (x) = f_{a+1}(x) \} = \{ x\in \Fbn : (_c D_{a+1}G \circ h) (x) = f_{a}(x) \}.
\end{split}
\end{equation*}

 Now we give a proof of the claim $(i)$. For any $a \in
\F_{2^{d}}$, it holds that $\cl(a) = 0$ and $\tr(a^{2^i+1}) =
\tr(a^2) = \tr(a)$ by Lemma \ref{2i+1lemma}. Therefore $V_a = \{x
\in \Fbn : \tr(a) = 0\}$ and $W_a = \{x \in \Fbn : \tr(a) = 1\}$,
which implies that one of $V_a$ and $W_a$ is $\Fbn$, and the other
is an empty set. Therefore one has either $ _c D_aG \circ h = f_a$
or $ _c D_aG \circ h = f_{a+1}$. Since both  $f_a$ and $f_{a+1}$
are bijective, $_c D_aG \circ h$  is also bijective, which
concludes the proof of the first claim.

Next we show the validity of the claim $(ii)$.
 By Proposition \ref{Jprop}, we can choose $x,y \in \Fbn$ and $a \in \Fbn$
 such that
 \begin{equation}\label{eqn8}
f_a(x) = f_{a+1}(y) \quad \textrm{ and }\quad
\tr\left((c+1)(x^{2^i+1}+y^{2^i+1})\right) = 1.
\end{equation}
Since $f_a(x) = f_{a+1}(y)$, one has
 $$(c+1)x^{2^i+1}
+ax^{2^i}+a^{2^i}x + a^{2^i+1} = (c+1)y^{2^i+1}
+(a+1)y^{2^i}+(a+1)^{2^i}y + (a+1)^{2^i+1}.$$
 By taking the absolute trace, we  obtain
\begin{equation*}
\begin{split}
\tr(&(c+1)x^{2^i+1} +ax^{2^i}+a^{2^i}x + a^{2^i+1} )  = \tr((c+1)y^{2^i+1} +(a+1)y^{2^i}+(a+1)^{2^i}y + (a+1)^{2^i+1}) \\
&\Leftrightarrow \tr((c+1)(x^{2^i+1} + y^{2^i+1}) )  = \tr( \cl(a)x + a^{2^i+1} + \cl(a)y + a^{2^i+1} + y^{2^i}+y + a^{2^i}  + a+1)\\
&\Leftrightarrow 1 = \tr((c+1)(x^{2^i+1} + y^{2^i+1}) )  =
\tr(\cl(a)x + a^{2^i+1}) + \tr(\cl(a)y + a^{2^i+1}).
\end{split}
\end{equation*}
Therefore one of $\tr( \cl(a)x + a^{2^i+1})$ and $\tr(\cl(a)y +
a^{2^i+1})$ is zero and the other is one, or equivalently,   one
of $x$ and $y$ is in $V_a$ and the other is in $W_a$ (Hence $a
\notin \F_{2^{d}}$.) If $x \in V_a$ and $y \in W_a$, then it holds
that
$$ (_c D_aG \circ h) (x) = f_a(x) = f_{a+1}(y) =  (_c D_aG \circ h) (y),  $$
 which implies that  $_c D_aG$ is not a permutation. Consequently, there exists $b\in \Fbn$ such that
    $_c \Delta_G(a,b) \geq 2$. If $x \in W_a$ and $y \in V_a$, then
    one has
 $$ (_c D_{a+1}G \circ h) (x) = f_a(x) = f_{a+1}(y) =  (_c D_{a+1}G \circ h) (y),  $$
 which means that $_c D_{a+1}G$ is not a permutation so that one concludes  $_c \Delta_G(a+1,b) \geq 2$ for some $b \in \Fbn$.
\end{proof}

Now we are ready to give a proof of Theorem \ref{maintheorem}.

\begin{proof}[Proof of Theorem \ref{maintheorem}:]
    \

    \smallskip
  Using Lemma \ref{threeimpossiblelemma} and the equation
  \eqref{4partpropre}, for all $a, b\in \Fbn$, one has
\begin{equation}
 _c \Delta_G(a,b) =
 \begin{cases}
  \,\, \Sum_{\bit_1, \bit_2 \in \F_2} \#S_{\bit_1\bit_2}(a,b) \leq 2 &\text{for } a \in T  \\
   \,\, \Sum_{\bit_1, \bit_2 \in \F_2} \#S'_{\bit_1\bit_2}(a,b) \leq 2 &\text{for } a \notin
   T,
\end{cases}  \\
\end{equation}
 which implies $_c \Delta_G\leq 2$. On the other hand, by Proposition \ref{geq2prop}, there exist $a \in
 \Fbn \setminus \F_{2^d}$ and $b \in \Fbn$ such
  that $_c \Delta_G(a,b) \geq 2$. Therefore one concludes that $_c \Delta_G= 2$ for any $c\in \F_{2^d}\setminus \F_2$.
\end{proof}

\section{More constructions of non-monomial PcN and APcN permutations on $\Fbn$}\label{sec4}
In this section, we explain some applications of the result of the
previous section and present more classes of non-monomial PcN and
APcN permutations related with $G$ and $h$.
\subsection{$c$-differential uniformity of the compositional inverse of $G$}

Though the differential uniformity is preserved under
compositional inversion due to the symmetry of the difference
equation, the
    situation for the $c$-differential uniformity is different, and
    there are plenty of permutations $F$ such that $_c\Delta_F\neq
    {}_c\Delta_{F^{-1}}$ where $F^{-1}$ is the compositional inverse of
    $F$.  However, for any monomial permutation $F(x)=x^r$ and its
    inverse $F^{-1}(x)=x^s$ on $\Fbn$ with $rs\equiv 1 \pmod{2^n-1}$,
    one has

\begin{lemma}\cite{WZ21}\label{power inverse lemma}\,\,
    $_c \Delta_F = \  _{c^{s}} \Delta_{F^{-1}}$ for $c \neq 0$.
\end{lemma}

\noindent It is well-known \cite{MRS+21} that $g(x)=x^{2^i+1}$ is
a PcN permutation for every $c\in \F_{2^d}\setminus \F_2$ in our
parameter setting ($n=dm,  i=d\ell, \gcd(m,2\ell)=1$ and $n>d$).
Let $g^{-1}(x)=x^u$ be the inverse permutation of $g(x)$, that is
$u \equiv (2^i+1)^{-1} \pmod{2^n-1}$. Recall that $2u\equiv
\frac{2^{im}+1}{2^i+1} \pmod{2^n-1}$ (See Lemma
\ref{2i+1lemma}$(iv)$).
\begin{corollary}\cite{WZ21}\label{pcnxu}
The monomials $x^u$ and $x^{2u}=x^{\frac{2^{im}+1}{2^i+1}}$ are
PcN permutations on $\Fbn$ for every $c\in \F_{2^d}\setminus
\F_2$, where $n=dm, i=d\ell$ with $n>d$ and $\gcd(m,2\ell)=1$.
\end{corollary}
\begin{proof}
By Lemma \ref{power inverse lemma},  PcNness of $g(x) = x^{2^i+1}$ implies
that  $g^{-1}(x) = x^{u}$ is P$c^{u}$N and $\{g^{-1}(x)\}^2 = x^{2u}$ is P$c^{2u}$N.
Since $\{ c : c \in \F_{2^{d}} \setminus \F_2\} =\{c^{u} : c\in \F_{2^{d}} \setminus \F_2\}
=\{c^{2u} : c\in \F_{2^{d}} \setminus \F_2\}$, we are done.
\end{proof}

\noindent Even though one can also use Corollary $11$ in \cite{WZ21} to derive above
result, we state as above since the above explicit form is necessary for our proof of
APcNness of $G^{-1}(x)$.

\smallskip

Now we are ready to explain the $c$-differential uniformity of the
compositional inverse of $G$.

\begin{prop}\label{inverseprop}
 Let $c\in \F_{2^d} \setminus \F_2$ with $d=\gcd(2i,n)=\gcd(i,n)<n$ and let $n$ be even.
  Let $u \equiv (2^i+1)^{-1} \pmod{2^n-1}$. Then the compositional inverse of
$G(x)$ is written as $G^{-1}(x) = x^u + \tr(x)$. Furthermore $_c
\Delta_{G^{-1}} \leq 2$ for $c \in \F_{2^{d}} \setminus \F_2$.
\end{prop}
\begin{proof}

Recall that $G =g \circ h$ where $h$ is the involution in  Lemma
\ref{2i+1lemma}. Then $G^{-1}(x) = h^{-1} \circ g^{-1}(x) = h
\circ g^{-1}(x) = h(x^u) = x^u + \tr(x).$ We now find the number
of solutions of $_cD_aG^{-1}(x) = b$. For $a,b \in \Fbn$, we have
\begin{equation*}
    \begin{split}
        G^{-1}(x+a) + c G^{-1}(x) = b \quad &\Leftrightarrow \quad (x+a)^u +\tr(x+a) + cx^u + c\tr(x) = b \\
        \quad  &\Leftrightarrow \quad (x+a)^u +  cx^u  = (c+1)\tr(x)+ b + \tr(a) \\
    \end{split}
\end{equation*}
Therefore we obtain
\begin{equation*}
    \begin{split}
        _c \Delta_{G^{-1}}(a,b) &= \# \{x : (x+a)^u +  cx^u  = (c+1)\tr(x)+ b + \tr(a)\} \\
        &= \# \{x : (x+a)^u +  cx^u  = b + \tr(a), \ \tr(x) = 0\} \\
        &\quad  + \# \{x : (x+a)^u +  cx^u  = b + \tr(a) + c+ 1,\ \tr(x) = 1\} \leq 1+ 1 = 2,
    \end{split}
\end{equation*}
where the last inequality comes from Corollary \ref{pcnxu}.
\end{proof}

\noindent  The above proposition says that $G^{-1}(x)$ is either
PcN or APcN for $c\in \F_{2^d}\setminus \F_2$. To show $_c
\Delta_{G^{-1}} = 2$, we need some preparations.

\begin{remark}
Note that in \cite{HPR+21},
 Hasan et al. gave a necessary condition for $F(x) = x^{2^k+1} + \gm\tr(x)$ to be PcN, but the complete
 conditions
  for which  $F(x)$ is PcN or APcN still remain open.
\end{remark}

\begin{lemma}\label{ulemma}
 Let $c\in \F_{2^d} \setminus \F_2$ with $d=\gcd(2i,n)=\gcd(i,n)<n$ and let $n$ be even.
Let  $u \equiv (2^i+1)^{-1} \pmod{2^n-1}$. Then,  one has the
followings:
\begin{enumerate}[(i)]
    \item $c^{2u} = c.$
    \item Letting $y= x^{2^i+1}$ (or $x=y^u$) with $x \in \Fbn$ and letting $a, b \in \Fbn$,  it holds that
    $$(x+a)^{2^i+1}+ cx^{2^i+1} = b \,\, \textit{ if and only if }\,\,
    \left(y+\frac{b}{c}\right)^u + c^{-u}y^u = c^{-u}a. $$

\end{enumerate}

\end{lemma}
\begin{proof}
For $(i)$, we have $c^{2u} = (c^2)^u = (c^{2^i+1})^u = c$
because
$c^{2^i} = c$. For $(ii)$, we obtain
\begin{equation*}
    \begin{split}
        (x+a)^{2^i+1}+ cx^{2^i+1} = b &\,\,\Leftrightarrow\,\, (x+a)^{2^i+1} + cy = b
        \,\,\Leftrightarrow\,\, x+a =  (cy + b)^u  \\
        &\,\,\Leftrightarrow\,\, (cy + b)^u + y^u  = a
        \,\,\Leftrightarrow\,\, (y + b/c)^u + c^{-u}y^u  =
        c^{-u}a,
    \end{split}
\end{equation*}
which completes the proof.
\end{proof}

To find $a,b \in \Fbn$ satisfying $_c \Delta_{G^{-1}}(a,b) = 2$,
we give the following proposition which corresponds to Proposition
\ref{Jprop}.
\begin{prop}\label{Jinvprop}
Let $c\in \F_{2^d} \setminus \F_2$ with $d=\gcd(2i,n)=\gcd(i,n)<n$
and let $n$ be even. Let  $u \equiv (2^i+1)^{-1} \pmod{2^n-1}$.
Then, there exist  $x,y \in \Fbn$ and $a \in \Fbn$ such that
$$(x+a)^u+cx^u = (y+a)^u + cy^u  +c+1  \,\,\textit{ and }\,\, \tr(x+y) = 1.
$$
\end{prop}
\begin{proof}
Let $ \tilde{c}=c^{-(2^i+1)}.$ Then, by Proposition \ref{Jprop},
we can find $\tilde{x},\tilde{y} \in \Fbn$ and $\tilde{a},\tilde{b} \in \Fbn$ such that
\begin{equation}\label{eqn9}
    (\tilde{x}+\tilde{a})^{2^i+1} + \tilde{c}\tilde{x}^{2^i+1} = \tilde{b} = (\tilde{y}+\tilde{a}+1)^{2^i+1} +  \tilde{c}\tilde{y}^{2^i+1}
\end{equation}
with $\tr(( \tilde{c}+1)(\tilde{x}^{2^i+1}+\tilde{y}^{2^i+1})) = 1$.
Letting $x_0=\tilde{x}^{2^i+1}$ and $y_0=\tilde{y}^{2^i+1}$ and using Lemma \ref{ulemma}-$(ii)$, the equation \eqref{eqn9}
can be rewritten as
\begin{equation}\label{eqn10}
    \left (x_0+\frac{\tilde{b}}{\tilde{c}}\right )^{u} + \tilde{c}^{-u}x_0^u = \tilde{c}^{-u}\tilde{a},
    \quad
    \left (y_0+\frac{\tilde{b}}{\tilde{c}}\right )^{u} + \tilde{c}^{-u}y_0^u = \tilde{c}^{-u}(\tilde{a}+1)
\end{equation}
with $\tr(( \tilde{c}+1)(x_0+y_0)) = 1$.
Also letting  $x=(\tilde{c}+1)x_0$ and $y=(\tilde{c}+1)y_0$, the equation \eqref{eqn10}
is written as
\begin{equation}\label{eqn11}
    \left (\frac{x}{\tilde{c}+1}+\frac{\tilde{b}}{\tilde{c}}\right )^{u} +
    \tilde{c}^{-u}\left(\frac{x}{\tilde{c}+1}\right)^u = \tilde{c}^{-u}\tilde{a},
    \quad
    \left (\frac{y}{\tilde{c}+1}+\frac{\tilde{b}}{\tilde{c}}\right )^{u} +
    \tilde{c}^{-u}\left(\frac{y}{\tilde{c}+1}\right) ^u = \tilde{c}^{-u}(\tilde{a}+1)
\end{equation}
with $\tr(x+y)=1$. Therefore, by multiplying $(\tilde{c}+1)^u$
to above equations, one has
\begin{equation}\label{eqn12}
    \left (x+\frac{\tilde{b}(\tilde{c}+1)}{\tilde{c}}\right )^{u} + \tilde{c}^{-u}x^u
    = \left( \frac{\tilde{c}+1}{\tilde{c}}\right )^u \tilde{a},
    \quad
    \left (y+\frac{\tilde{b}(\tilde{c}+1)}{\tilde{c}}\right )^{u} + \tilde{c}^{-u}y^u
    = \left( \frac{\tilde{c}+1}{\tilde{c}}\right )^u(\tilde{a}+1)
\end{equation}
with $\tr(x+y) = 1$. Note that $\tilde{c}=c^{-(2^i+1)}$
(i.e., $c = \tilde{c}^{-u}$)
implies that
$$ \left( \frac{\tilde{c}+1}{\tilde{c}}\right )^u  =
\left(1+ \tilde{c}^{-1}\right)^u = \left(1+ c^{2^i+1}\right)^u =
\left(1+ c^2\right)^u =  \left(1+ c\right)^{2u} =1 + c,$$
where the last equality comes from Lemma \ref{ulemma}-$(i)$.
Finally, letting $a= \frac{\tilde{b}(\tilde{c}+1)}{\tilde{c}}$ and
$b=(c+1)\tilde{a}$, the equation \eqref{eqn12} can be written as
$$ (x+ a)^u + c x^u  = b, \quad  (y+a)^u + cy^u = b + c+ 1 \quad \text{ with } \tr(x+y)= 1,$$
which becomes
$$
(x+ a)^u + c x^u = (y+a)^u + cy^u + c+ 1 \text{ and } \tr(x+y)=  1.
$$
\end{proof}

Now we are ready to give a proof of $_c \Delta_{G^{-1}} = 2$.

\begin{theorem}\label{inversemaintheorem}
Let $c\in \F_{2^d} \setminus \F_2$ with $d=\gcd(2i,n)=\gcd(i,n)<n$
and let $n=dm$ be even. Let  $u \equiv (2^i+1)^{-1} \pmod{2^n-1}$.
Then,  $G^{-1}(x)=x^u+\tr(x)$ is an APcN permutation on $\Fbn$. Or
equivalently, $x^\frac{2^{im}+1}{2^i+1}+\tr(x)$ is an APcN
permutation on $\Fbn$.
\end{theorem}
\begin{proof}
Let $c\in \F_{2^d} \setminus \F_2$ be given.  Then the proof of Proposition \ref{inverseprop} says that, for any $a$ and $b$ in $\Fbn$, one has
\begin{equation}\label{eqn13}
    \begin{split}
        _c \Delta_{G^{-1}}(a,b) =\,  &\# \{x\in \Fbn : (x+a)^u +  cx^u  = b + \tr(a), \ \tr(x) = 0 \} \\
        &+ \# \{ x \in \Fbn: (x+a)^u +  cx^u  = b + \tr(a) + c+ 1,  \ \tr(x) =
        1\}\leq 1+1=2.
    \end{split}
\end{equation}
Also, by Proposition \ref{Jinvprop}, there exist $x,y \in \Fbn$ and $a \in \Fbn$ satisfying
$$(x+a)^u+cx^u = (y+a)^u + cy^u  +c+1  \,\,\textit{ and }\,\, \tr(x+y) = 1.$$
Since $\tr(x+y) = 1$, one has either $\tr(x)=0$ and $\tr(y)=1$
or $\tr(x)=1$ and $\tr(y)=0$. For the first case, by defining $b\in \Fbn$ as
$$(x+a)^u+cx^u =b+\tr(a)= (y+a)^u + cy^u  +c+1, $$
the inequality in  \eqref{eqn13} becomes an equality. Similarly
for the second case, by defining $b\in \Fbn$ as
$$(x+a)^u+cx^u =b+\tr(a)+c+1= (y+a)^u + cy^u  +c+1, $$
the inequality in  \eqref{eqn13} becomes an equality.
\end{proof}

\subsection{Generalization of $h$ and $G$ and their $c$-differential uniformities}
Recall that, when $n$ is even, $h(x)=x+\tr(x^{2^i+1})$ defined in
Lemma \ref{2i+1lemma} is an involution on $\Fbn$. Moreover, it
will turn out that $h$ is PcN when $c\in \F_{2^d}\setminus \F_2$
and is APcN when $c\in \F_{2^n}\setminus \F_{2^d}$, where
$d=\gcd(n,2i)$ is even. However, we will prove the above statement
in a more general setting where both even and odd $n$ are
considered. So, from now on, we assume that $n$ is any positive
integer $>1$.

Let $\alpha \in \Fbn$ and let $i$ be a positive integer. Define a
polynomial
\begin{equation}\label{halpha}
    h_{\alpha}(x) = x + \tr(\alpha x + x^{2^i+1})
\end{equation}
 over $\Fbn$. Note that, if $\alpha=0$, then $h_\alpha(x) =h(x)$.
The following lemma gives a condition for which $h_\alpha$ becomes
an involution.
\begin{lemma}\label{halp involution lemma}
 If $\tr(\alpha+1)=0$ (i.e., if $n$ is even and $\tr(\alpha)=0$ or if $n$ is odd and $\tr(\alpha)=1$),
 then $h_\alpha(x)$ is an involution on $\Fbn$.
\end{lemma}
\begin{proof}
    Using $(x+ \tr(\alpha x + x^{2^i+1}))^{2^i+1} = x^{2^i+1}+ (x^{2^i}+x+1) \tr(\alpha x + x^{2^i+1}) $, one has
    \begin{equation*}
        \begin{split}
            &h_\alpha(h_\alpha(x)) \\
            &= x+ \tr(\alpha x + x^{2^i+1}) +\tr\left (\alpha\left(x+ \tr(\alpha x + x^{2^i+1})\right) + \left(x+ \tr(\alpha x + x^{2^i+1})\right)^{2^i+1} \right )\\
            &= x + \tr(\alpha x + x^{2^i+1})+ \tr\left(\alpha x + \alpha\tr(\alpha x + x^{2^i+1})\right) + \tr\left(x^{2^i+1} + (x^{2^i} + x+ 1)\tr(\alpha x + x^{2^i+1})\right)\\
            &= x + \tr(\alpha x + x^{2^i+1})+ \tr(\alpha x) + \tr(\alpha)\tr(\alpha x + x^{2^i+1}) + \tr(x^{2^i+1}) + \tr(x^{2^i} + x+ 1)\tr(\alpha x + x^{2^i+1})\\
            &= x + \tr(\alpha+x^{2^i}+x+1)\tr(\alpha x + x^{2^i+1})= x + \tr(\alpha+1)\tr(\alpha x + x^{2^i+1})=x.
        \end{split}
    \end{equation*}
\end{proof}


\begin{theorem}\label{thmfinal}
 Let $\alpha \in \Fbn $ with $\tr(\alpha+1)=0$ and  let $d = \gcd(n,2i)$. Then, for the
 involution  $ h_{\alpha}(x) = x + \tr(\alpha x + x^{2^i+1})$,   one has
    the followings.
    \begin{enumerate}[$(i)$]
        \item If $c \in \F_{2^{d}} \setminus \F_2$, then $h_{\alpha}(x)$ is
        PcN on $\Fbn$.
        \item If $c \in \F_{2^{n}} \setminus \F_{2^{d}}$, then $h_{\alpha}(x)$ is
        APcN on $\Fbn$.
    \end{enumerate}
    In particular, for extreme cases of $d$,
          \begin{enumerate}
        \item[$\ast$]  if $d=n$ {\rm (}e.g., $n=2i${\rm )}, $h_\alpha(x)= x + \tr(\alpha x +
    x^{2^{n/2}+1})$ is PcN for all $c \in \F_{2^{n}} \setminus
    \F_2$.
          \item[$\ast\ast$] if $d=1$ {\rm (}i.e., $\gcd(n, 2i)=1${\rm )}, $h_\alpha(x)= x + \tr(\alpha x +
    x^{2^{i}+1})$ is APcN for all $c \in \F_{2^{n}} \setminus
    \F_2$.
          \end{enumerate}
\end{theorem}
\begin{proof}
    For $c \notin \F_2$, one has
    \begin{equation}\label{eqn20}
        \begin{split}
            {}_c D_ah_{\alpha}(x) &= (x+a) + \tr(\alpha(x+a) + (x+a)^{2^i+1}) +c( x +\tr(\alpha x + x^{2^i+1}))\\
            &= (1+c)(x+\tr(\alpha x + x^{2^i+1})) + \tr((x+a)^{2^i+1} + x^{2^i+1})+a +\tr(\alpha a)\\
            &= (1+c)h_{\alpha}(x) +\left ( \tr(\cl(a)x) + \tr(\alpha a +a^{2^i+1})\right
            )+a,
        \end{split}
    \end{equation}
    where $\cl(a)=a^{2^i}+a^{2^{-i}}$ was introduced in Lemma
    $\ref{2i+1lemma}$.  In a similar way as in the equation \eqref{eqn3},
    we have
        $\Fbn = V_{a,\alpha} \cup W_{a,\alpha}$
    where
      \begin{equation}\label{eqnVWalp}
        \begin{split}
      V_{a,\alpha} & = \{x \in \Fbn : \tr(\cl(a)x) = \tr(\alpha a+ a^{2^i+1})\}, \\
       W_{a,\alpha} & = \{x \in \Fbn : \tr(\cl(a)x) =
    \tr(\alpha a+a^{2^i+1})+1\}.
       \end{split}
    \end{equation}
     From \eqref{eqn20}, one has
    \begin{equation}\label{main function3}
        {}_c D_ah_{\alpha}(x) = \begin{cases}  (c+1)h_{\alpha}(x) +a &\text{ if }  x \in V_{a,\alpha}\\
            (c+1)h_{\alpha}(x) +1+a &\text{ if }  x \in W_{a,\alpha}
        \end{cases}
    \end{equation}
    and thus ${}_c D_ah_{\alpha}$ is one to one on each subset $V_{a,\alpha}$ and
    $W_{a,\alpha}$ since $h_{\alpha}$ is one to one. From this observation, one easily gets
    $_{c} \Delta _{h_{\alpha}}\leq 2$ because
      \begin{align*}
      _{c} \Delta _{h_{\alpha}}(a,b) &= \#\{x\in \Fbn : {}_c D_a{h_{\alpha}}(x)=b\} \\
             &= \#\{x\in V_{a,\alpha} : {}_c D_ah_{\alpha}(x)=b \}+\#\{x\in W_{a,\alpha} : {}_c
D_ah_{\alpha}(x)=b \}\leq
    2.
       \end{align*}
    Therefore, if there exist $x\neq y \in \Fbn$ satisfying
    \begin{equation}\label{eqn19-2}
        {}_c D_a{h_{\alpha}}(x) =  {}_c D_ah_{\alpha}(y),
    \end{equation}
    then one has either  $x\in V_{a,\alpha}, y \in W_{a,\alpha}$ or $x\in W_{a,\alpha}, y \in
    V_{a,\alpha}$.
    We now suppose that $x \in V_{a,\alpha}$ and  $y \in W_{a,\alpha}$ without loss of
    generality. Then, letting $J = (\cl(a)x + \alpha a +a^{2^i+1}) + (\cl(a)y
    +\alpha a +a^{2^i+1})$, one gets $\tr(J)=0+1=1$. Also, the equation
    $\eqref{eqn19-2}$ becomes
    \begin{equation*}
        (c+1)h_{\alpha}(x) = (c+1)h_{\alpha}(y)+ 1.
    \end{equation*}
    Since $h_{\alpha}$ is an involution,  letting $c'=\frac{1}{c+1}$, the above
    equation is rewritten as
    \begin{equation*}
        \begin{split}
            h_{\alpha}(y) &= h_{\alpha}(x) + c' \,\,\   \,\,\,\\
            \Leftrightarrow      y &= h_{\alpha}\left(h_{\alpha}(x) + c'\right) \\
            &= h_{\alpha}(x) + c'+ \tr\left(\alpha\left(h_{\alpha}(x) + c'\right) +\left(h_{\alpha}(x) + c'\right)^{2^i+1} \right)\\
            &= h_{\alpha}(x)  + \tr\left( \alpha h_{\alpha}(x) + h_{\alpha}(x)^{2^i+1} \right) + \tr\left(h_{\alpha}(x)^{2^i}c' + h_{\alpha}(x){c'^{2^i}} + c'^{2^i+1} \right) + c' + \tr(\alpha c')\\
            &=  h_{\alpha}(h_{\alpha}(x))  + \epsilon + c' + \tr\left( \alpha c'  + c'^{2^i+1} \right)  = x + \epsilon +
            h_{\alpha}\left(c'\right),
        \end{split}
    \end{equation*}
    where $\epsilon = \tr\left(h_\alpha(x)^{2^i}c' + h_\alpha(x)c'^{2^i} \right)$.
    Therefore one has  $x + y = \epsilon + h_{\alpha}\left(c'\right)$, where $x
    \in V_{a,\alpha}$ and  $y \in W_{a,\alpha}$ satisfy \eqref{eqn19-2}. Then, using
    $\tr(J)=1$ and $x+y=\epsilon + h_{\alpha}\left(c'\right)$, one has
    \begin{equation*}\begin{split}
            1=\tr(J) = \tr(\cl(a) (x+y)) &= \tr\left (a^{2^i} (x+y)+ a (x+y)^{2^i}\right ) \\
            &= \tr\left (a^{2^i}( \epsilon + h_{\alpha}(c')) + a (\epsilon + h_{\alpha}(c'))^{2^i}\right )\\
            &= \tr\left ( a^{2^i} h_{\alpha}(c') + a  h_{\alpha}(c')^{2^i} + a^{2^i}\epsilon+ a \epsilon^{2^i}\right )\\
            &=   \tr\left (
            \cl(h_{\alpha}(c'))a \right ) \,\,\, (\because \epsilon \in \F_2).
        \end{split}
    \end{equation*}
    Now, for the claim $(i)$,  suppose that $c \in \F_{2^{d}}
    \setminus \F_2$ where $d=\gcd(n,2i)$.
     Then it holds that $h_{\alpha}(c') = c' + \tr(\alpha c' + c'^{2^i+1})
    \in \F_{2^{d}}$. Since $\ker \cl =\F_{2^d}$ by Lemma \ref{2i+1lemma} $(iii)$, one
    has $\cl(h_{\alpha}(c')) = 0$ and consequently $\tr(J) = 0$, which is a
    contradiction. Thus, there exist no $x\neq y$ satisfying
    $\eqref{eqn19-2}$ if $c\in \F_{2^d}\setminus \F_2$,  and one
    concludes that ${}_cD_ah_{\alpha}$ is a permutation for every $a\in \Fbn$.
    For the claim $(ii)$, suppose that $c \in \F_{2^{n}} \setminus
    \F_{2^{d}}$. (So $c' \in \Fbn \setminus \F_{2^d}$ and $h_\alpha(c')\in \Fbn \setminus \F_{2^d}$.)
     Since $\cl(h_{\alpha}(c'))\neq 0$ again by Lemma
    \ref{2i+1lemma} $(iii)$, there exists $a \in \Fbn$ satisfying
    $\tr(J) = \tr\left ( \cl(h_{\alpha}(c'))a \right ) = 1$. Hence for such $a$
    and given $x \in V_{a,\alpha}$, the equation \eqref{eqn19-2} has two
    solutions $x \neq y$ with $y \overset{\rm{def}}{=} x + \epsilon +
    h_{\alpha}\left(c'\right) \in W_{a,\alpha}$. (Note that $\epsilon + h_{\alpha}\left(c'\right) \neq 0$
    because $ h_{\alpha}\left(c'\right) \notin \F_{2^{d}}$.)
\end{proof}


\begin{remark}
Theorem \ref{thmfinal} generalizes Theorem 4 in \cite{HPS+22},
where Hasan, Pal and St\v{a}nic\v{a} showed the followings: If
$n>1$ is odd and  $\tr(\alpha)=1$ with $\gcd(n,i)=1$, then one has
$_c \Delta _{h_\alpha} \leq 2$ and $c$-DDT entry at $(a,b)$ is
given by
\begin{equation*}
_c \Delta _{h_\alpha}(a,b)  = \begin{cases}
0  &\text{ if }  A=1 \text{ and } B=1 \\
1 &\text{ if }   B=0  \\
 2 &\text{ if }  A=0 \text{ and } B=1,
\end{cases}
\end{equation*}
where $ A=\tr\left( \frac{(a+b)(a^{2^i}+a^{2^{-i}})}{1+c}+\alpha
a+a^{2^i+1}\right)$ and $B=\tr\left(
\frac{a^{2^i}+a^{2^{-i}}}{1+c} \right). $ The assumption that $n$
is odd and $\gcd(n,i)=1$  in \cite{HPS+22} is a special case of
our Theorem \ref{thmfinal} where $d=\gcd(n,2i)=1$. According to
Theorem \ref{thmfinal},  one always has $_c \Delta _{h_\alpha}=2$
for $c\in \Fbn\setminus \F_2$ when $d=1$, i.e., When $d=1$,
$h_\alpha$ cannot be PcN. In other words, there must exist $a,
b\in \Fbn$ satisfying $A=0$ and $B=1$.

\end{remark}

\noindent Letting $\alpha = 0$ and $n=\textrm{even}$ (i.e.,
$\tr(\alpha+1)=\tr(1)=0)$, one also has the following corollary.
\begin{corollary}
    Let $d = \gcd(n,2i)$ with $n$ even and $h(x) = x + \tr(x^{2^i+1})$. Then
    \begin{enumerate}[$(i)$]
        \item If $c \in \F_{2^{d}} \setminus \F_2$ then one has $_c \Delta _{h} = 1$.
        \item If $c \in \F_{2^{n}} \setminus \F_{2^{d}}$ then one has $_{c} \Delta _{h} = 2$.
    \end{enumerate}
\end{corollary}

\noindent Using $h_\alpha$, one may also define $G_\alpha$ as
$$
G_\alpha(x)=g\circ h_\alpha(x)=\left(x+\tr(\alpha x
+x^{2^i+1})\right)^{2^i+1}.
$$
Then, under the same conditions of Theorem \ref{thmfinal} and
Lemma \ref{2i+1lemma} $(i)$, one can prove in a similar way as in
Section \ref{sec3} that $G_\alpha$ is also an APcN permutation on
$\Fbn$. That is, using the following modifications,

\begin{enumerate}[(a)]
\item $T \rightarrow T_\alpha \overset{\rm{def}}{=}\{x\in \Fbn :
\tr(\alpha x+x^{2^i+1})=0\}$ in Lemma \ref{Tclosinglemma},
 \item $\tr(x^{2^i+1})\rightarrow
\tr(\alpha x+x^{2^i+1})$ in Lemma \ref{Tclosinglemma} and in the
definitions of $S_{\bit_1\bit_2}, S'_{\bit_1\bit_2}$ in eq.
\eqref{Sabdefine},
 \item  $ V_a \rightarrow V_{a,\alpha}, \,\,\,
   W_a \rightarrow W_{a,\alpha}$ in eq. \eqref{VaWadef},
\end{enumerate}

\noindent one has the following theorem.

\begin{theorem}\label{maintheorem22}
   %


{\rm (}Generalization of Theorem \ref{maintheorem}{\rm )} Let
$\alpha \in \Fbn$ with $\tr(\alpha+1)=0$. Suppose that
$d=\gcd(2i,n)=\gcd(i,n)<n$. {\rm (}i.e., $\frac{n}{d}>1$ is odd
with $d=\gcd(i,n)$.{\rm )}  Then,  for $c\in \F_{2^d}\setminus
\F_2$, $G_\alpha(x) = \left(x+ \tr(\alpha x +
x^{2^i+1})\right)^{2^i+1}$ is an APcN permutation on $\Fbn$.
\end{theorem}

\begin{remark} Unlike Theorem \ref{inversemaintheorem} where we
showed that $G^{-1}$ is also APcN, we only find that $_c
\Delta_{G_\alpha^{-1}}\leq 4$  and $G_\alpha^{-1}$ is not APcN in
general when $\alpha\neq 0$. We briefly explain why the
compositional inverse of $G_{\alpha}$ may not be APcN. Since
$G_{\alpha}^{-1}(x) =h_\alpha^{-1}\circ g^{-1}(x)=h_\alpha\circ
g^{-1}(x) =h_{\alpha}(x^{u})$ with $u \equiv (2^i+1)^{-1}
\pmod{2^n-1}$, one has
$$G_{\alpha}^{-1}(x) = x^{u} + \tr(x + \alpha x^{u} ).$$
Then we have
\begin{equation*}
    \begin{split}
        G_{\alpha}^{-1}(x+a) + c G_{\alpha}^{-1}(x) = b \,\, &\Leftrightarrow \,\, (x+a)^u +\tr(x+a+\alpha(x+a)^u) + cx^u + c\tr(x+\alpha x^u) = b \\
        \,\,  &\Leftrightarrow \,\, (x+a)^u +  cx^u  = \tr(x+a+\alpha(x+a)^u) + c\tr(x+\alpha x^u) + b  \\
    \end{split}
\end{equation*}
and thus, by letting $f_1(x)= \tr(x+a+\alpha(x+a)^u)$ and
$f_2(x)=\tr(x+\alpha x^u)$, one has
\begin{equation*}
    \begin{split}
        _c \Delta_{G_\alpha^{-1}}(a,b) =\,  &\# \{x\in \Fbn : (x+a)^u +  cx^u  = b, \ f_1(x) = 0, f_2(x) = 0 \} \\
        &+ \# \{ x \in \Fbn: (x+a)^u +  cx^u  = c+b,  \ f_1(x) = 0, f_2(x) = 1\} \\
         &+ \# \{ x \in \Fbn: (x+a)^u +  cx^u  = 1+b,  \ f_1(x) = 1, f_2(x) = 0\}\\
          &+ \# \{ x \in \Fbn: (x+a)^u +  cx^u  =1+ c+b,  \ f_1(x) = 1, f_2(x) = 1\}\leq
          4,
    \end{split}
\end{equation*}
where  $_c \Delta_{x^u}=1$ is used in the last inequality.
Therefore unlike the equation \eqref{eqn13} where $\alpha=0$, we
cannot guarantee that $G_\alpha^{-1}$ is APcN. Sage experiment
shows that $3\leq {}_c \Delta_{G_\alpha^{-1}}\leq 4$ when
$\alpha\neq 0$.
\end{remark}

\begin{remark}\label{pante}
    Very recently, in \cite{LRS22}, Li, Riera and St\v{a}nic\v{a} showed that  $_{c} \Delta _{H_k}\leq 2$  for $c
    \notin \F_2$,
    where $H_k(x) = L(x) + \tr(x^{2^k+1})$
    and $L(x)$ is a $2$-linearized permutation with $L(1)=1$.
\end{remark}

\section{Conclusion}
Permutations are very useful in the design of invertible S-boxes.
However, not many classes of APcN permutations are reported so
far, and most of known APcN permutations are monomials. In this
paper, we presented new  APcN permutations on $\Fbn$,
$$G(x) = \left(x+ \tr(x^{2^i+1})\right)^{2^i+1}, \quad G^{-1}(x)
= x^u+ \tr(x) \,\,\left(G^{-1}(x)^2 = x^\frac{2^{im}+1}{2^i+1}+
\tr(x)\right)$$ for all  $c\in \F_{2^d}\setminus \F_2$, where
$n=dm, i=d\ell, \gcd(n,i)=d,$ $d$ is even and $m>1$ is odd. Though
APcN property is not preserved under inversion in general, we
proved that both $G$ and $G^{-1}$ are APcN permutations for $c\in
\F_{2^d}\setminus \F_2$. We also proved that, if $\tr(\alpha
+1)=0$ with  $d=\gcd(n,2i)$, the involution
$h_\alpha(x)=x+\tr(\alpha x+ x^{2^i+1})$ is PcN for $c\in
\F_{2^d}\setminus \F_2$, and is APcN for $c\in \F_{2^n}\setminus
\F_{2^d}$. Moreover, we showed that $G_\alpha(x) = \left(x+
\tr(\alpha x + x^{2^i+1})\right)^{2^i+1}$ is APcN  for  $c\in
\F_{2^d}\setminus \F_2$, where $d=\gcd(2i,n)=\gcd(i,n)<n$ and
$\tr(\alpha+1)=0$.

\begin{appendix}

    \section{Proof of Lemma \ref{trnrlemma} $(i)$}\label{appendix}
    Let $d = 2^st$ where $s \geq1$ and $t$ is odd.  We use induction on $s$. When
    $s=1$ (i.e., $d = 2t$), any element of $\F_{2^{d}}^{\times}$ is uniquely written as
    $ux$, where $u \in U \overset{\textrm{def}}{=} \{u \in \F_{2^{2t}} : u^{2^t+1}=1\}$
    and $x \in  \F_{2^{t}}^{\times}$ because $\gcd(\# U, \#
    \F_{2^d}^\times)=\gcd(2^t+1,2^t-1)=1$.
     Since $3 \nmid 2^t - 1$, it holds that $\F_{2^{t}}^{\times} \subset H=\{x^3 : x\in \F_{2^d}^\times\}$,$\text{ so that }$
     \begin{equation*}
        \zeta \F_{2^{t}}^\times + \zeta^2 \F_{2^{t}}^\times \subset \zeta H + \zeta^2 H.
    \end{equation*} By the assumption  $\zeta \notin H$,
    the fact $\F_{2^{t}}^{\times} \subset H$ implies that $\zeta \notin \F_{2^{t}}$.
     Then $\{\zeta,\zeta^2\}$ is a basis for the vector space $\F_{2^{d}}$ over $\F_{2^{t}}$.
     Since $\#\left(    \zeta \F_{2^{t}}^\times + \zeta^2 \F_{2^{t}}^\times\right) = (2^t-1)^2 = (2^{2t}-1)-2(2^t-1)
      = \#\F_{2^{2t}}^\times - 2(2^t-1) $ and
      since $\F_{2^{2t}}^\times= \left( \zeta \F_{2^{t}}^\times + \zeta^2 \F_{2^{t}}^\times\right) \cup
       \zeta \F_{2^{t}}^\times \cup \zeta^2 \F_{2^{t}}^\times$ is
       a disjoint union,
       to show $\F_{2^{d}}^\times = \zeta H + \zeta^2 H$,
      we need to show that $ \zeta \F_{2^{t}}^\times \subset \zeta H + \zeta^2 H$ and
       $ \zeta^2 \F_{2^{t}}^\times \subset \zeta H + \zeta^2 H$. We choose any $u^3 \in U_3
       \overset{\textrm{def}}{=} \{u^3 : u \in U\} \subset H.$
       Then $\{\zeta u^3 , \zeta^2 u^3\}$ is also a basis for $\F_{2^{2t}}$ over $\F_{2^{t}}$. Also it holds that
         $\zeta \F_{2^{t}}^\times \cap  \zeta u^3 \F_{2^{t}}^\times =\emptyset$ and
         $\zeta^2 \F_{2^{t}}^\times \cap  \zeta u^3 \F_{2^{t}}^\times=\emptyset$ because
         $u\notin \F_{2^{t}}^\times$ and $\zeta \notin H$. Therefore we can obtain
\begin{equation*}
        \zeta\F_{2^{t}}^\times \subset \zeta u^3 \F_{2^{t}}^\times + \zeta^2 u^3 \F_{2^{t}}^\times \subset   \zeta u^3 H + \zeta^2 u^3 H
        = \zeta  H + \zeta^2 H,  \,\,\, \textrm{ and similarly }
        \,\,  \zeta^2\F_{2^{t}}^\times \subset \zeta  H + \zeta^2 H
\end{equation*}
which implies $\F_{2^{2t}}^{\times} = \zeta  H + \zeta^2 H$.

Now suppose that the lemma is true for all $j$ with $1\leq j \leq
s$ where $d = 2^st$ with $t$ odd. We now show that the lemma is
also true for $\F_{2^{2^{s+1}t}} = \F_{2^{2d}}$. Since
$\gcd(2^d-1,2^d+1) = 1$, one also has
 $\F_{2^{2d}}^{\times} = \F_{2^{d}}^{\times} \cdot U  \overset{\textrm{def}}{=} \{xu : x\in \F_{2^d}^\times, u \in U \}$
 where $U =
\{ u \in \F_{2^{2d}} : u^{2^d+1} = 1\}.$ Let $H = \{x^3 : x \in
\F_{2^{2d}}^\times \}$ and $H' = \{x^3 : x \in \F_{2^{d}}^\times
\}$.  Since $d$ is even (i.e., $3 \nmid 2^d+1$), one has $U =U_3=
\{u^3 : u \in U\}$ and thus $H = H'\cdot U$. Therefore any $\zeta
\notin H = H'\cdot U$ is uniquely written as $ \zeta = \zeta'u$
where $\zeta' \in \F_{2^{d}}^\times \setminus H'$ and $u \in U$.
By induction hypothesis, one has $\F_{2^{d}}^\times = \zeta' H' +
\zeta'^2 H' $, and
\begin{equation*}
    \begin{split}
        \F_{2^{2d}}^\times = \F_{2^{d}}^\times \cdot U &= ( \zeta' H' + \zeta'^2 H') \cdot U \\
        &=\zeta' H'U + \zeta'^2 H'U \\
        &= \zeta'u H'U + \zeta'^2u^2 H'U = \zeta H + \zeta^2H,
    \end{split}
\end{equation*}
hence the proof is complete. \qed
\end{appendix}

\end{document}